\documentclass[runningheads]{llncs}

\usepackage[utf8]{inputenc}
\usepackage[pdftex,dvipsnames]{xcolor}

\usepackage{tabularx}

\usepackage{todonotes}

\usepackage{mathtools}
\usepackage{amssymb}
\usepackage{xspace}
\usepackage{paralist}
\usepackage{dsfont}
\usepackage[ligature,reserved,inference]{semantic}

\usepackage{tikz-cd}

\usepackage{url}

\usepackage{amsmath}
\usepackage{amsfonts}
\usepackage{amssymb}
\usepackage{mathtools}
\usepackage{proof}
\usepackage{colortbl}

\usepackage{hyperref}
\usepackage{nameref}
\usepackage{wrapfig}

\usepackage{listings}
\usepackage[ruled]{algorithm}
\usepackage[noend]{algpseudocode}

\usepackage{environ}

 \newcommand{\textbfup}[1]{\textup{\texttt{#1}}}
 \reservestyle{\mech}{\textbfup}
 \mech{first, fail} 

\newcommand{\ff}{\textit{ff}}
\newcommand{\cf}{\textit{cf}}
\newcommand{\cc}{\textit{cc}}
\newcommand{\fc}{\textit{fc}}
\newcommand{\ffc}{\textit{ffc}}
\newcommand{\cfc}{\textit{cfc}}
\newcommand{\ccc}{\textit{ccc}}
\newcommand{\fcc}{\textit{fcc}}

\newcommand{\cff}{\textit{cff}}

\newtheorem{specification}{Specification}
\newcommand{\myparagraph}[1]{\paragraph{\textup{\textbf{#1}}}}

\newcommand{\fnfhookup}{Fail/NoFail\xspace}

\mathlig{<-}{\leftarrow}
\mathlig{==}{\equiv}
\mathlig{++}{\mathop{+\!\!+}}
\mathlig{-x}{\mathbin{-\mspace{-8mu}\times}}
\mathlig{->>}{\twoheadrightarrow}

\reservestyle{\variables}{\text}
\variables{depth[height],
paths[leaves], tx}

\reservestyle{\mathfunc}{\mathsf}
\mathfunc{applytx[applyTx], basicTx}
\mathfunc{prune[decide], leafs[allFutures], nextTx, add}
\mathfunc{extend[attach],  innerprune[prune], step, allFutures, 
succ, timeout, failmap, successors, undecided, timeoutCommit, failmapCommit, failmapFail,
allMonitoringCommitWithTimeout[knownToCommitWithTimeout], commitWithTimeout,
monitoringContracts, allMonitoringCommit[knownToCommit],
oneMonitoringFail[knownToFail], allFutureAllMonitoringCommitWithTimeout}

\reservestyle{\stmt}{\textbf}
\stmt{return, assert}

\newcommand\nextTx{\<nextTx>}

\newcommand\add{\<add>}
\newcommand\suc{\<succ>}
\newcommand\leaves{\<leafs>}
\newcommand\leafs{\<leafs>}
\newcommand\allFutures{\<allFutures>}

\newcommand\applyTx{\<applytx>}
\newcommand\applyTxOld{\<basicTx>}

\newcommand\depth{\<depth>}
\newcommand\timeout{\ensuremath\mathsf{timeout}}
\newcommand\txid{\texttt{txid}}

\newcommand\fail{\mathsf{Fail}}
\newcommand\nofail{\mathsf{Commit}}

\newcommand\commitType{\ensuremath\mathbf{Commit}}
\newcommand\failType{\ensuremath\mathbf{Fail}}
\newcommand\pendingType{\ensuremath\mathbf{Pending}}

\newcommand\undefined{\ensuremath{\texttt{\textbf{?}}}}
\newcommand\None{\mathsf{None}}

\algnewcommand\algorithmicswitch{\textbf{switch}}
\algnewcommand\algorithmiccase{\textbf{case}}
\algnewcommand\algorithmicassert{\texttt{assert}}
\algnewcommand\Assert[1]{\State \algorithmicassert(#1)}%
\algdef{SE}[SWITCH]{Switch}{EndSwitch}[1]{\algorithmicswitch\ #1\ \algorithmicdo}{\algorithmicend\ \algorithmicswitch}%
\algdef{SE}[CASE]{Case}{EndCase}[1]{\algorithmiccase\ #1}{\algorithmicend\ \algorithmiccase}%
\algtext*{EndSwitch}%
\algtext*{EndCase}%

\newcommand{\repeattheorem}[1]{%
  \begingroup
  \renewcommand{\thetheorem}{\ref{#1}}%
  \expandafter\expandafter\expandafter\theorem
  \csname reptheorem@#1\endcsname
  \endtheorem
  \endgroup
}

\NewEnviron{reptheorem}[1]{%
  \global\expandafter\xdef\csname reptheorem@#1\endcsname{%
    \unexpanded\expandafter{\BODY}%
  }%
  \expandafter\theorem\BODY\unskip\label{#1}\endtheorem
}

\NewEnviron{replemma}[1]{%
  \global\expandafter\xdef\csname replemma@#1\endcsname{%
    \unexpanded\expandafter{\BODY}%
  }%
  \expandafter\lemma\BODY\unskip\label{#1}\endlemma
}

\newcommand{\repeatatlemma}[1]{%
  \begingroup
  \renewcommand{\thelemma}{\ref{#1}}%
  \expandafter\expandafter\expandafter\lemma
  \csname replemma@#1\endcsname
  \endlemma
  \endgroup
}

\usepackage{listings}
\usepackage[dvipsnames]{xcolor}

\definecolor{verylightgray}{rgb}{.97,.97,.97}

\lstdefinelanguage{Solidity}{
	keywords=[1]{anonymous, assembly, assert, balance, break, call, callcode, case, catch, class, constant, continue, constructor, contract, debugger, default, delegatecall, delete, do, else, emit, event, experimental, export, external, false, finally, for, function, gas, if, implements, import, in, indexed, instanceof, interface, internal, is, length, library, log0, log1, log2, log3, log4, memory, modifier, new, payable, pragma, private, protected, public, monitor_storage, pure, push, require, return, returns, revert, selfdestruct, send, solidity, storage, struct, suicide, super, switch, then, this, throw, transfer, true, try, typeof, using, value, view, while, with, addmod, ecrecover, keccak256, mulmod, ripemd160, sha256, sha3}, 
	keywordstyle=[1]\color{blue}\bfseries,
	keywords=[2]{set,address, bool, byte, bytes, bytes1, bytes2,
	bytes3, bytes4, bytes5, bytes6, bytes7, bytes8, bytes9,
	bytes10, bytes11, bytes12, bytes13, bytes14, bytes15, bytes16,
	bytes17, bytes18, bytes19, bytes20, bytes21, bytes22, bytes23,
	bytes24, bytes25, bytes26, bytes27, bytes28, bytes29, bytes30,
	bytes31, bytes32, enum, int, int8, int16, int24, int32, int40,
	int48, int56, int64, int72, int80, int88, int96, int104,
	int112, int120, int128, int136, int144, int152, int160,
	int168, int176, int184, int192, int200, int208, int216,
	int224, int232, int240, int248, int256, mapping, string, elem,
	uint, uint8, uint16, uint24, uint32, uint40, uint48, uint56,
	uint64, uint72, uint80, uint88, uint96, uint104, uint112,
	uint120, uint128, uint136, uint144, uint152, uint160, uint168,
	uint176, uint184, uint192, uint200, uint208, uint216, uint224,
	uint232, uint240, uint248, uint256, var, void, ether, finney,
	szabo, wei, days, hours, minutes, seconds, weeks, years, map,
        txId},	
	keywordstyle=[2]\color{teal}\bfseries,
	keywords=[3]{block, blockhash, coinbase, difficulty, gaslimit,
	number, timestamp, msg, gas, sender, sig, value, now, tx,
	gasprice, origin, add, epochinc, get, setminus, emptyset,
	system, first, queue, ustore, isSelfOnly, gen_op,
	verdict, tx_id, preCommit, abort, preAbort},	
	keywordstyle=[3]\color{violet}\bfseries,
	keywords=[4]{init, term, begin,end, upon, timeout, monitor,
	failmap, COMMIT, UNDECIDED,FAIL,fail},
	keywordstyle=[4]\color{BrickRed}\bfseries,
	identifierstyle=\color{black},
	sensitive=false,
	comment=[l]{//},
	morecomment=[s]{/*}{*/},
	commentstyle=\color{gray}\ttfamily,
	stringstyle=\color{red}\ttfamily,
	morestring=[b]',
	morestring=[b]"
}

\newcommand{\Thanks}{\thanks{This work was funded in part by PRODIGY Project
   (TED2021-132464B-I00)---funded by MCIN/AEI/10.13039/501100011033/ and
   the European Union NextGenerationEU/PRTR---by DECO Project
   (PID2022-138072OB-I00)---funded by MCIN/AEI/10.13039/501100011033 and
   by the ESF+---and by a research grant from Nomadic Labs and the Tezos
   Foundation.}
 }
\title{Monitoring the Future of Smart Contracts\Thanks}

\def\orcid#1{\smash{\href{http://orcid.org/#1}{\protect\raisebox{-1.25pt}{\protect\includegraphics{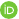}}}}}

\author{Margarita Capretto\inst{1,2}\orcid{0000-0003-2329-3769}\and
  Martin Ceresa\inst{1}\orcid{0000-0003-4691-5831} \and
  C\'esar S\'anchez\inst{1}\orcid{0000-0003-3927-4773}
}
\authorrunning{Capretto, Ceresa and Sánchez}
\institute{IMDEA Software Institute, Spain \and
  Universidad Politécnica de Madrid (UPM), Madrid, Spain}

\begin{document}

\maketitle

\begin{abstract}
  Blockchains are decentralized systems that provide trustable
  execution guarantees.
  Smart contracts are programs written in specialized programming languages
  running on blockchains that govern how tokens and cryptocurrency are sent and
  received.
  Smart contracts can invoke other smart contracts during the execution of
  transactions always initiated by external users.

  Once deployed, smart contracts cannot be modified, so techniques
  like runtime verification are very appealing for improving their
  reliability.
  However, the conventional model of computation of smart contracts is
  transactional: once operations commit, their effects are permanent
  and cannot be undone.

  In this paper, we proposed the concept of \emph{future monitors} which allows
monitors to remain waiting for future transactions to occur
before committing or aborting.
  This is inspired by optimistic rollups, which are modern blockchain
  implementations that increase efficiency (and reduce cost) by
  delaying transaction effects.
  We exploit this delay to propose a model of computation that allows
  (bounded) future monitors.
  We show our monitors correct respect of legacy
  transactions, how they implement future bounded monitors and how
  they guarantee progress.
  We illustrate the use of future bounded monitors to
  implement correctly multi-transaction flash loans.
\end{abstract}


\section{Introduction}\label{sec:introduction}
%
%
\emph{Blockchains}~\cite{nakamoto06bitcoin} were first introduced as
distributed infrastructures that eliminate the need of trustable third
parties in electronic payment systems.
Modern blockchains incorporate smart
contracts~\cite{szabo96smart,wood2014ethereum}~(contracts hereon),
which are stateful programs stored in the blockchain that govern the
functionality of blockchain transactions.
Users interact with blockchains by invoking
contracts,
whose execution controls the exchange of
cryptocurrency.
Contracts allow sophisticated functionality, enabling many
applications in decentralized finances (DeFi), decentralized
governance, Web3, etc.

%
%
Contracts are written in high-level programming languages, like
Solidity~\cite{solidity} and Ligo~\cite{alfour20ligo}, which are then
typically compiled into low-level bytecode languages like
EVM~\cite{wood2014ethereum} or Michelson~\cite{michelson}.
Even though contracts are typically small compared to
conventional software, writing contracts is notoriously
difficult.
The open nature of the invocation system---where every contract can
invoke every other contract---facilitates that malicious users break
programmer's assumptions and steal user tokens
(e.g.~\cite{Daian.2016.DAO}).
%
%
Once installed, contract code is immutable,
and the effect of
running a contract cannot be reverted (the contract \emph{is} the
law).
%

%
%
Two classic reliability approaches can be applied to contracts:
\begin{compactitem}
\item \textbf{static techniques} ranging from static
  analysis~\cite{Stephens.2021.Smartpulse} and model
  checking~\cite{Permenev.2020.VerX} to deductive software
  verification
  techniques~\cite{ahrendt20functional,nehai19deductive,bhargavan16formal,conchon19verifying},
  theorem
  proving assistants~\cite{Bruno.2019.MiChoCoq,annenkov20concert,schiffl20formal}
  or assisted formal construction of
  programs~\cite{Sergey.2018.Scilla}.
\item \textbf{dynamic
    verification}~\cite{ellul18runtime,azzopardi18monitoring,li20securing,capretto22monitoring}
  dynamically inspecting the execution of contracts against 
specifications taking corrective measures.
\end{compactitem}
%
%
We follow in this paper a dynamic monitoring technique.
Monitors are a defensive mechanism to express desired properties that
must hold during the execution of the contracts.
If the property fails, the monitor fails the whole transaction.
Otherwise, the execution finishes normally according to the contract
code.
In practice, monitors are mixed within the contract code, which limits
the properties that can be monitored.
In~\cite{capretto22monitoring}, the authors presented a hierarchy of
monitors, including operation and transaction monitors.
An operation monitor for a contract $A$ runs alongside \(A\) and reads
and modifies specific monitor variables stored in
$A$~\cite{ellul18runtime,azzopardi18monitoring,li20securing}.
Operation monitors can only execute when $A$ is invoked and cannot
inspect or invoke other contracts.
Transaction monitors~\cite{capretto22monitoring} can inspect
information across a full transaction, even after the last invocation
of $A$ in the transaction.
For example, the return of a loan \emph{within the transaction} is an
important property that can be monitored with a transaction monitor
and not by an operation monitor, because a transaction must fail if
the money lent is not returned by the end of the transaction.

Traditional blockchain systems cannot implement transaction
monitors~\cite{capretto22monitoring}, but fortunately, this is easy to
achieve by extending the execution model with two simple features: a
\<first> instruction and a \fnfhookup hookup mechanism.
Instruction \<first> returns \emph{true} during the first invocation
of the contract in the current transaction.
The \fnfhookup mechanism equips each contract with a new flag, \<fail>, that can
be assigned (to \emph{true} or \emph{false}) during the execution of the
contract (and that is \emph{false} by default).
The semantics of \<fail> is that transactions fail if at least one  contract
has its \<fail> flag set to \emph{true} at the end of the transaction.

In this paper, we study an even richer notion of monitors that enables
to fail or commit depending on \emph{future
  transactions}.
Future monitors can predicate on sequences of transactions during a
bounded period of time.
This period of time, called the \emph{monitoring window} is fixed
a priori.

\myparagraph{Optimistic rollups.}
Future monitors can be incorporated easily in Layer-2 Optimistic
Rollups\footnote{optimistic rollups for short}, which are an approach
to improve blockchain scalability by moving computation and data
off-chain.
The most popular optimistic rollup implementation is
Arbitrum~\cite{ArbitrumNitro},
implemented on top of the Ethereum blockchain~\cite{wood2014ethereum}.
%
%
%
%
Arbitrum offers the same API as Ethereum, allowing to install and
invoke Ethereum contracts.
Arbitrum transactions are executed off-chain and their effects are
submitted as \emph{assertions}.
%
%
Assertions are \textit{optimistically} assumed to be correct and a
fraud-prove arbitration scheme allows to detect invalid assertions.
Assertions are pending during a challenging period\footnote{Currently
a week.}
to allow observers to check their correctness.
The arbitration game consists of a bisection protocol, played between
the challenger and asserter, which has the property that the honest
player can always win the dispute.
Assertions that survive until the end of the challenge period become
permanent.
Future monitors can exploit the delay imposed by the challenging
period to fail or commit based on information from the future.

\myparagraph{Bounded Future Monitoring.}
In this article, we enrich transaction monitors with a controlled
ability to predicate about the future evolution of blockchains.
Contracts are extended to include: \txid, \(\<failmap>\), and
\(\<timeout>\).
The instruction \txid\ returns the (unique) current transaction
identifier.
Each contract is equipped with a map $\<failmap>$ indicating---for
each transaction involving the contract---whether the future monitor
of the transaction is activated or not, and if so, its monitoring
status (commit, fail or undecided).
By default, future monitoring is deactivated.
%
%
Contracts can modify their $\<failmap>$ (1) to activate the future
monitor of the current transaction, or (2) to commit or fail undecided
future monitors of previous transactions within the monitoring window.
If a contract sets a past transaction \(\<failmap>\) entry to fail,
the corresponding transaction fails.
%
%
The \(\<timeout>\) function is invoked at the end of the monitoring
window to decide whether to fail or commit if the future monitor of
the transaction is still undecided.
This guarantees that transactions cannot be pending after a bounded
amount of time.

We call our monitors \emph{future monitors} since the decision to
commit or fail may depend on transactions that will execute in the
future.
Future monitors expand the monitor hierarchy presented
in~\cite{capretto22monitoring}, which included operation and
transaction monitors as well as monitors that involve several
contracts (multicontract monitors) or even the whole blockchain
(global monitors), but always in the context of a single transaction.
When combined with future monitors, we obtain \emph{multicontract
  future monitors} and \emph{global future monitors}, but we leave
these extensions as future work.
A particular subclass of \emph{multicontract future monitors} was
studied in~\cite{ellul21optional} focusing on long-lived
transactions~\cite{gray81transaction}, whose lifetime span blockchain
transactions and potentially involve different contracts and parties.
Fig.~\ref{fig:hierarchy} shows the updated monitoring hierarchy
including future monitors.
\begin{figure}[h!]
  \begin{tabular}{|l|l|}\hline
    \textbf{Present} & \textbf{Future} 
    \\\hline
    Global monitors  & Global future monitors [future work] 
    \\\hline
    Multicontract monitors  & Multicontract future monitors \cite{ellul21optional} [future work]
    \\\hline
    Transaction monitors \cite{capretto22monitoring} &\cellcolor{gray!35} Future monitors [this work] 
    \\\hline 
    Operation monitors
    \cite{azzopardi18monitoring,ellul18runtime,li20securing}   \\\cline{1-1}
  \end{tabular}
  \caption{Monitor hierarchy. The first column belongs to~\cite{capretto22monitoring}.}
  \label{fig:hierarchy}
\end{figure}

The rest of the paper is organized as follows.
Section~\ref{sec:prelims} presents an abstract model of computation,
which is extended in Section~\ref{sec:modelOfComputation} to support
future monitors.
In Section~\ref{sec:properties}, we discuss their properties.
Section~\ref{sec:example} shows a real-world example of future
monitors.
In Section~\ref{sec:related-work} we discuss related work, and
Section~\ref{sec:conclusion} concludes.


\section{Model of Computation}\label{sec:prelims}

We introduce now our abstract model of computation to reason about
blockchains.

\myparagraph{Blockchains Execution Overview.}
Blockchains are incremental permanent records of executed transactions
packed in blocks.
Transactions are in turn composed of a sequence of operations where
the initial operation is an invocation from an external user.
Each operation invokes a destination contract, which is identified by
its unique address.
The execution of an operation follows the instructions of the program
(the contract) stored at the destination address.
Contracts can modify their local storage and invoke other
contracts.

Transaction execution consists of executing operations, computing
their effects (which may include the generation of new operations)
until either (1) there are no more pending operations, or (2) an
operation fails or the available gas is exhausted.
In the former case, the transaction commits and all changes are made
permanent.
In the latter case, the transaction fails and no effect takes place in
the storage of contracts, except that some gas is consumed.
Therefore, the state of contracts is determined by the effects of
committing transactions.
%

\myparagraph{Model of Computation.}
Our model computation describes blockchain state evolution as the
result of sequential transaction executions.
Blockchain configurations are records containing all information
required to compute transactions, such as: a partial map between
addresses and their storage and balance, plus additional information
about the blockchain such as block number.
We use $\Sigma$ to denote blockchain configurations and
\(\mathcal{U}\) to denote balances of external users.

Transactions are the result of executing a sequence of operations
starting from an external operation placed by a user.
Transactions can either commit, if every operation is successful, or
fail, if one of its operations fails or the gas is exhausted.
We use function $\applyTxOld$, which takes a transaction, a blockchain
configuration, and balances of external users as inputs, and returns
the blockchain configuration and the external user balances that
result from executing the transaction in the input configuration.
Additionally, predicate $\suc$ indicates whether the execution of a
transaction commits or fails in a given blockchain configuration and
external user balances.
Furthermore, function $\mathsf{discount}$ deducts the specified amount
of tokens from the balance of the indicated user in the provided
external user balances.
The following relation $\leadsto_{tx}$ defines the evolution of the
blockchain using $\applyTxOld$, $\suc$ and $\mathsf{discount}$:
\[
  \begin{tabular}{lcr}
  \inference[\texttt{commit}]
    {\applyTxOld(tx,\Sigma,\mathcal{U}) = (\Sigma', \mathcal{U}') \\ \suc(tx,\Sigma, \mathcal{U}) = \text{commit}}
    {\Sigma, \mathcal{U} \leadsto_{tx} \Sigma', \mathcal{U}'}
    & \hspace{0.5em} &
 \inference[\texttt{fail}]
    {
    \mathcal{U}' = \mathsf{discount}(\mathcal{U}, \mathsf{src}(tx),\mathsf{cost}(tx))
    \\ \suc(tx,\Sigma,\mathcal{U}) = \text{fail}
    }
      {\Sigma, \mathcal{U} \leadsto_{tx} \Sigma, \mathcal{U}'}
  \end{tabular}
\]

If a transaction fails~(rule~\texttt{fail}), the blockchain configuration is
preserved, but the external user originating the transaction pays for
the resources consumed.
Cost and resource analysis are out of the scope of this paper, so we
ignore the computation of \(\mathcal{U}\).

Operation and transaction monitors are defined at the operation and
transaction level, and thus, they are implemented inside $\applyTxOld$
and abstracted away in this model.


\section{Bounded Future Monitored Blockchains}\label{sec:modelOfComputation}
In this section, we present a modified model of computation supporting
future monitors.
The main addition is the implementation of monitoring transactions
predicating on future transactions within a monitoring window $k$.
The monitoring window captures for how long (in the number of
transactions) the monitor can predicate on.
This additional feature enables us to install a monitor per transaction.
Future instances of contracts that activated a future monitor can
decide to either fail or commit the past transaction within the
monitoring window.
If any contract sets to fail the transaction future monitor of a past
transaction, the monitored transaction fails.
Otherwise, when all contracts that monitor a given transaction commit
the transaction becomes permanently committed.

\subsection{Future \(k\)-bounded Monitors}\label{sec:futurekmonitors}
Transactions can commit or fail depending on their subsequent
\(k\) transactions, and thus, the post-state after executing a
transaction may depend on future transactions.
At any given point in time, transaction future monitors may:
\begin{compactitem}
\item fail because at least one contract involved set the monitor to
  fail;
\item commit because all contracts involved set the monitor to commit;
\item stay pending.
\end{compactitem}
Therefore, we identify three transaction monitor states: known to fail, (denoted
by \(\fail\)), known to commit (denoted by \(\nofail\)) and undecided
(denoted by \(\undefined\)).
Finally, we add another value to represent transactions without monitors:
$\None$.

\paragraph*{Failing Map.}
A contract $C$ can only interact with the future monitor of
transaction $t$ if $C$ was involved in $t$.
To keep track of different monitors for $C$ (for different
transactions), every contract $C$ has a map, called \emph{failing
  map}, from transactions to monitor states.

At the start of a transaction, the monitor is deactivated and can only
be activated during the current transaction.
Therefore, if at the end of a transaction $t$ no contract updated the
failing map of its monitor for $t$, then the behavior is like legacy
unmonitored transactions (as previously described in
Section~\ref{sec:prelims}).

A contract $C$ can modify its failing map many times but only the entries
of
\begin{wrapfigure}[7]{r}{0.37\linewidth}
  \vspace{-1.3em}
  \centering
  \includegraphics[scale=0.37]{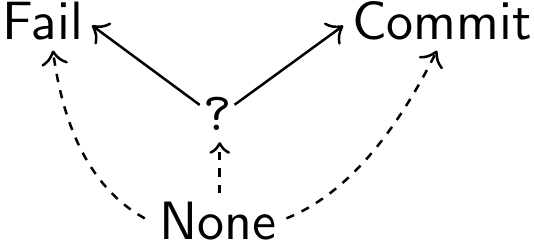}
  \caption{Monitor transitions.}
  \label{fig:states}
\end{wrapfigure}
those transactions where $C$ was involved and activated the monitor.
Changes to failing maps at the end of transactions can be (1) the
activation of the monitor for the current transaction (from $\None$ to
$\fail, \nofail$, or $\undefined$, indicated by dashed arrows in
Fig.~\ref{fig:states}); or (2) decisions reached for undecided
monitors (from $\undefined$ to $\fail$ or $\nofail$, indicated by
plain arrows).

\paragraph*{Timeout.}
Contracts have a new special function called $\timeout$ that can be
used to describe the decision of undecided monitors at the monitoring
window.
Function $\timeout$ takes a transaction identifier and returns either
$\fail$ or $\nofail$ and it is set by contracts.
The default $\timeout$ function returns $\nofail$.

At the end of the monitor window, the system invokes $\timeout$
if the failing map entry for that transaction is marked as
$\undefined$.
If at least one contract involved in the transaction decides to fail,
the transaction fails, and otherwise the transaction commits.

\subsection{Extending the Model of Computation}~\label{sec:extmodelcomp}
%
%
We extend blockchain configurations with a \emph{future monitor context} $\Delta$
associating contracts with their failing map and $\timeout$ function.

\paragraph{Transaction Execution:}
Transactions can immediately commit or fail, or depend on future
transactions that happen within the monitoring window, so the
execution of a transaction can return one of the following cases:
\begin{compactitem}
\item a new configuration as an immediate commit,
\item a new configuration as an immediate fail,
\item two \emph{possible} new configurations, one for failing and one
  for committing, which depends on the future.
\end{compactitem}
These behaviors are captured by a new function $\applyTx$ that checks
if future monitors were activated during the transaction.
Future monitors restrict the behavior of the blockchain, because they
only modify the blockchain evolution making transactions fail more
often.

Non-monitored transactions either immediately commit or fail based on
function $\suc$, and their effects are equivalent to the traditional
model.

The function $\applyTx$, when applied to a monitored not failing
transaction, returns two blockchain configurations, describing the
only two possible futures.
The first configuration represents the effects if the transaction
commits, and the second represents a failing transaction, so in these
cases the post-configurations are identical to the previous
configurations (modulo resources consumed).

A contract $C$ can only modify its failing map to activate the future
monitor of the current transaction or to decide future monitors that
$C$ had previously activated but not yet decided.
If a contract incorrectly updates its failing map, the current
transaction fails.
When transactions fail, the system does not modify any failmap map or
timeout function.

\paragraph{Blockchain System.}
There are two types of transactions: \emph{permanent} (committed or
failed) and \emph{pending} transactions.
%
%
%
\emph{Blockchain runs} are pairs $(H,\tau)$ consisting of a sequence
$H$ of consolidated blockchain configurations called the
\emph{history} and a directed tree $\tau$ where each internal node has
one or two children.
$H$ contains only permanent transaction.
Tree $\tau$ is called the \emph{monitoring tree} and includes pending
transactions.
Each node in the monitoring tree is a blockchain configuration.
The monitoring tree represents all possible sequences of blockchain
states that the list of pending transactions can generate.
Exactly one path in the tree will eventually survive and become part
of $H$, which depends on whether the corresponding transactions commit
of fail.
Each level in the tree corresponds to the execution of transactions up
to that level but different configuration at the same level is a
different possible reality.
To simplify notation, we use $n$ to refer to the blockchain
configuration captured by node $n$ in the tree.
The root of the monitoring tree is the last blockchain configuration
that was consolidated, that is, the last blockchain configuration in
the history sequence.

The height of the monitoring tree is at most \(k\).
It can be shorter than \(k\) at the genesis of the blockchain but once
the first \(k\) transactions have been executed the monitoring tree
reaches and maintains a height \(k\).
In the worst case, depending on the contracts deployed in the
blockchain, the monitoring tree can have \(2^{k+1}-1\) nodes, but in
general not every transaction is going to be monitored which reduces
the branching and hence the size of the tree.

Fig.~\ref{fig:blockchain_run} shows a blockchain run $(H,\tau)$.
The first $j+1$ transactions are permanent and the
last $k$ transactions are pending.
The last permanent blockchain configuration is $(\Sigma,\Delta)$
and it is also the root of the monitoring tree \(\tau\).
When the first pending transaction, $t_{j+1}$, executes from
configuration $(\Sigma,\Delta)$, a contract $C$ that executed in
\(t_{j+1}\) activated the transaction future monitor generating a
branching in $\tau$.
However, not all transactions generate a branching in the monitoring
tree as not all transactions are necessarily monitored, (for example
$t_{j+k}$).
Configuration $(\Sigma',\Delta')$ is a one of the possible outcomes of
executing all pending operations.

\begin{figure*}[t!]
    \centering
    \includegraphics[scale=0.4]{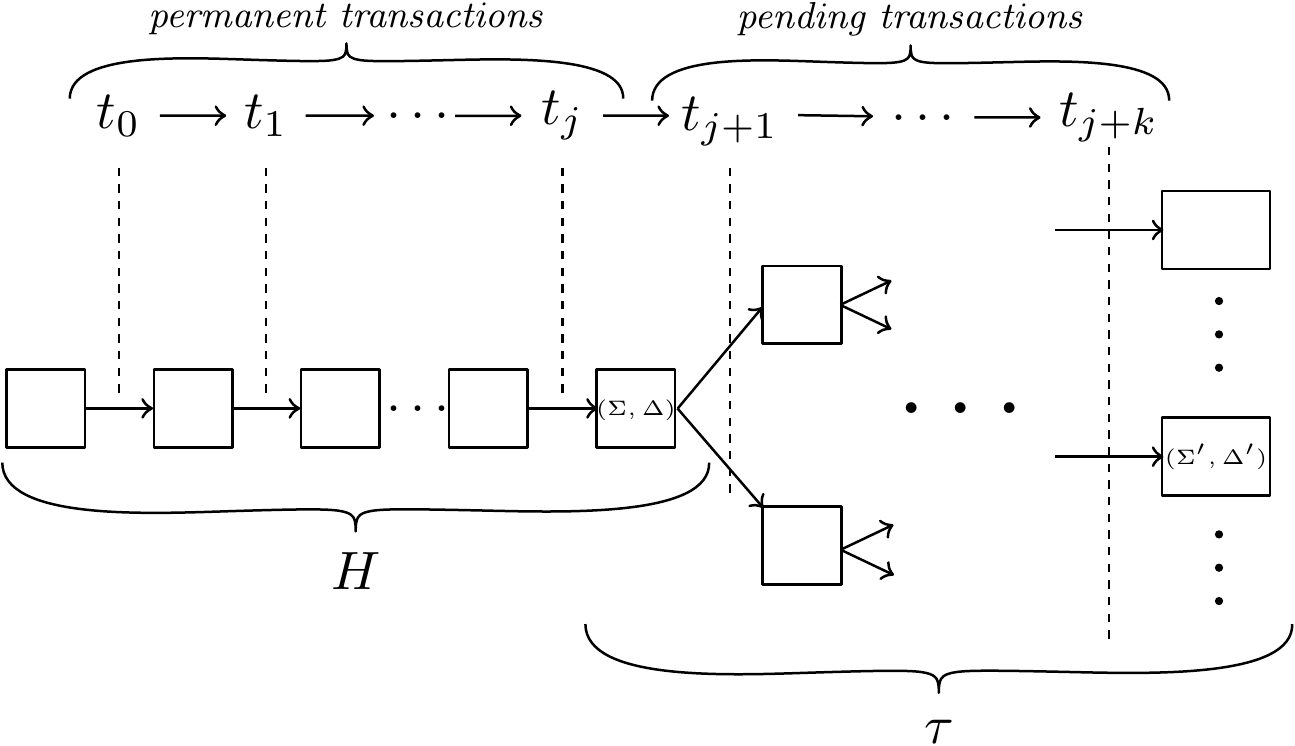}
    \caption{A blockchain run of $j+1$ permanent
      transactions and $k$ pending transactions.}
    \vspace{-2em}
    \label{fig:blockchain_run}
\end{figure*}

\paragraph*{Notation.}
We use the following functions:
\begin{compactitem}
\item \(\nextTx(n)\): returns the transaction that labels the outgoing
  edges from $n$.
\item The \emph{successor} of a node \(n \xrightarrow{t} n'\) in the
  monitoring tree.
\item \(\<successors>(n)\): given a node $n$ that is not a leaf,  returns all
successors of $n$, which can be $(n_c,n_f)$, where $n_c$ is the committing
successor and $n_f$ the failing successor, or $n'$ if $n$ is not
branching.
\item the \emph{committing subtree} of $n$: the maximal subtree rooted at
  the committing successor of $n$.
\item the \emph{failing subtree} of $n$: the maximal subtree rooted at
  the failing successor of $n$.
\item $\allFutures(n)$: the set of leaves reachable from node $n$.

\end{compactitem}
Consider \(n \xrightarrow{t} n'\).
The configuration at \(n'\) is one of the possible results of
executing transaction \(t\) from the blockchain configuration at
\(n\).
For simplicity, when referring to a monitoring tree \(\tau\) with the root
node \(n\), we use the terms \(\tau\) and \(n\) interchangeably. 
Thus, \(\<successors>(\tau)\) denotes the successors of the root node
of \(\tau\).
The possible futures of the root node of monitoring tree $\tau$,
denoted by $\allFutures(\tau)$, is referred as the futures in
$\tau$.

\begin{example}
  The following figure shows an example run after $7$ transactions, starting
  at initial blockchain configuration $N_0$ and monitoring
  window $k=2$.

  \begin{center}
    \includegraphics[scale=0.4]{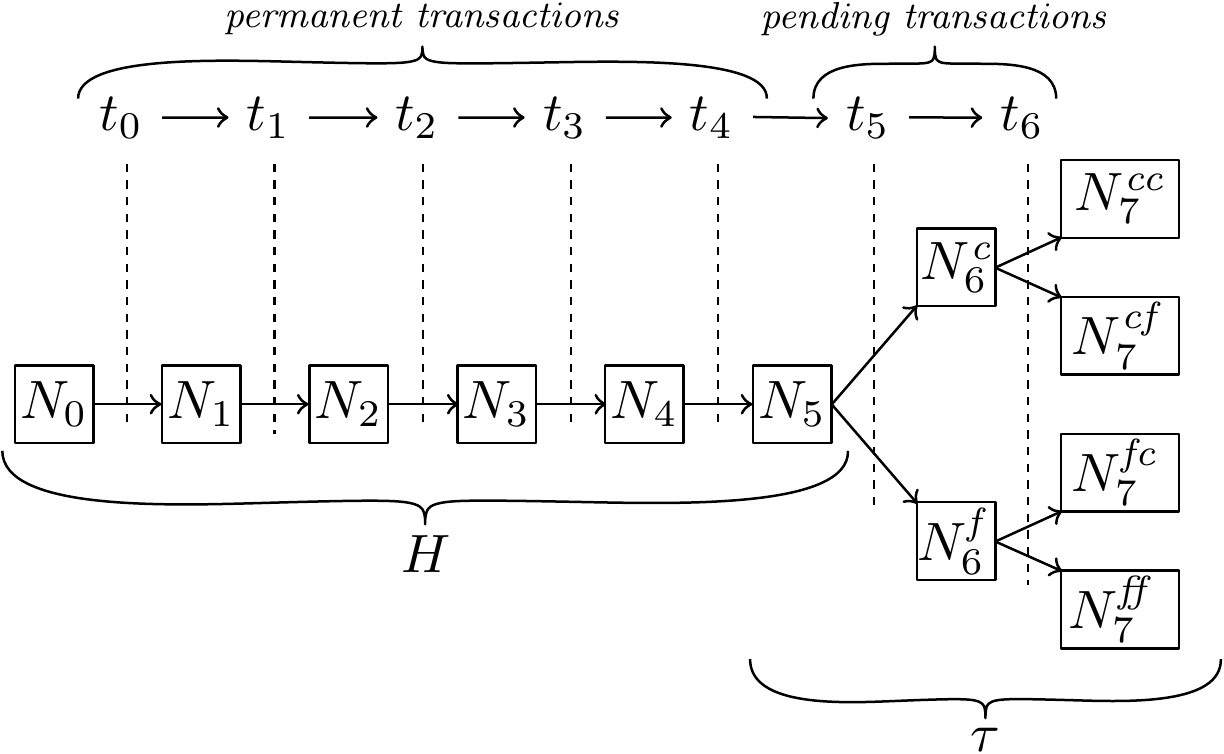}
  \end{center}

  \noindent History \(H\) corresponds to the first $5$ permanent
  transactions.
  The remaining transactions are pending forming a directed tree
  \(\tau\) whose root is $N_5$.
  The transaction at node \(N_5\) is \(\<nextTx>(N_5) = t_5\).
  Node \(N_5\) successors are
  \(\<successors>(N_5) = (N^{c}_6,N^{f}_6)\).
  The committing subtree of \(N_5\) is the subtree with root
  \(N^{c}_6\) and the failing subtree of \(N_5\) is the subtree with
  root \(N^{f}_6\).
  Finally, the futures in \(\tau\) are
  $\allFutures(\tau) =
  \{N^{\cc}_7,N^{\cf}_7,N^{\fc}_7,N^{\ff}_7\}$.
  We annotate with superscript \(c\) and \(f\) the committing and
  failing transactions, respectively, and group them in sequences
  describing paths in monitoring trees. 
\end{example}


\subsection{Blockchain evolution}~\label{sec:blockevo}
The evolution of the blockchain is defined by function \(\<step>\)
(see Fig.~\ref{fig:alg:stepextend}) which takes blockchain runs and
transactions, and extends runs.
The system has only one rule:
\[
  \begin{tabular}{c}
  \inference
      {\<step>((H,\tau), t) = (H',\tau')}
    {(H,\tau) ->>_{t} (H',\tau')}
  \end{tabular}
\]
Valid traces are defined by the relation $->>$ and consist of chains
of related blockchain states
$(H_0,\tau_0) ->>_{t_0} (H_1,\tau_1) ->>_{t_1} \ldots$ where
$(H_0,\tau_0)$ is an initial blockchain run with
$\tau_0 = H_0 = (\Sigma,\Delta)$.

\newcommand{\StepBody}{
    \renewcommand{\thealgorithm}{Step}         
   \begin{algorithmic}
     \Function{$\<step>$}{$(H,\tau), t$}
       \State $\tau' \leftarrow \<extend>(\tau,t)$
       \If{ \(\depth(\tau') \leq k\)} 
         \State $\<return> (H,\tau')$
       \Else
         \State $\tau'' \leftarrow \<prune>(\tau')$
         \State $\<tx> \leftarrow \nextTx(\tau)$
         \State $H.add( \tau \xrightarrow{\<tx>} \tau'')$
         \State $\<return> (H, \tau'')$
       \EndIf  
      \EndFunction
    \end{algorithmic}
}


\newcommand{\ExtendBody}{
    \renewcommand{\thealgorithm}{Attach}         
   \small
   \begin{algorithmic}
     \Function{$\<extend>$}{$\tau, t$} \Comment{\small Extends
       monitoring trees.}
       \State $\tau' \leftarrow \tau$
       \For{$l \in \leaves(\tau)$}
         \Switch{$\applyTx(t,l)$}
           \Case{$\commitType(l_c): \, \tau'.\add(l \xrightarrow{t} l_c)$ } 
           \EndCase
           \Case{$\failType(l_f): \, \tau'.\add(l \xrightarrow{t} l_f)$} 
           \EndCase
           \Case{$\pendingType(l_c,l_f){: \, \tau'.\add(l \xrightarrow{t} (l_c,l_f))}$}
           \EndCase  
         \EndSwitch 
       \EndFor
       \State \<return> $\tau'$  
     \EndFunction
   \end{algorithmic}
}


\newcommand{\PruneBody}{
    \renewcommand{\thealgorithm}{Prune}         
   \small
   \begin{algorithmic}
     \Function{$\<prune>$}{$\tau$} \Comment{Decides commit/fail of
       the root transaction of $\tau$}
       \State \<assert> $\depth(\tau) = k+1$          
       \State $\tau' \leftarrow \<innerprune>(\tau)$ 
       \State \(t \leftarrow \nextTx(\tau)\)
       \Switch{\(\<successors>(\tau')\)}
          \Case{\(\tau''\): \<return> \(\tau''\)}
          \EndCase
          \Case{\((\tau_c, \tau_f)\):}
            \If{\(\forall l \in \leafs(\tau_c): \<allMonitoringCommitWithTimeout>(l,t)\)}
              \(\!\!\,\)
              \<return> \(\tau_c\)
            \Else
              \(\!\,\)
              \<return> \(\tau_f\)
            \EndIf
          \EndCase
        \EndSwitch  
     \EndFunction
   \end{algorithmic}
}
\newcommand{\FigPrune}{
  \begin{figure}[h!]
    \PruneBody
\caption{\small Function decide}
 \label{fig:alg:prune}
\end{figure}
}


\newcommand{\InnerPruneBody}{
   \small
   \begin{algorithmic}
     \Function{$\<innerprune>$}{$\tau$} 
       \If{\(\tau\) is a leaf} 
         \<return> \(\tau\)
       \EndIf
       \State \(t \leftarrow \nextTx(\tau)\)
       \Switch{\(\<successors>(\tau)\)}
          \Case{\(\tau'\): \<return> \(\tau \xrightarrow{t} \<innerprune>(\tau')\)}
          \EndCase
           \Case{\((\tau_c, \tau_f)\):}
             \State \(\tau'_c \leftarrow \<innerprune>(\tau_c)\)
             \State \(\tau'_f \leftarrow \<innerprune>(\tau_f)\)
             \If{\(\forall l \in \leafs(\tau'_c): \<allMonitoringCommit>(l,t)\)}
               \<return> \(\tau \xrightarrow{t}\tau'_c\)
            \EndIf
            \If{\(\forall l \in \leafs(\tau'_c): \<oneMonitoringFail>(l,t)\)}
              \<return> \(\tau \xrightarrow{t}\tau'_f\)
             \EndIf
             \State \<return> \(\tau \xrightarrow{t} (\tau'_c,\tau'_f)\) 
           \EndCase
        \EndSwitch  
     \EndFunction
    \end{algorithmic}
    \vspace{-1em}
  }
  
\newcommand{\FigInnerPrune}{
  \begin{figure}[h!]
    \InnerPruneBody
    \caption{Function innerprune: removes impossible inner nodes.}
     \label{fig:alg:innerprune}
\end{figure}
}


\newcommand{\AuxBody}{
   
   \begin{algorithmic}
     \Function{$\<failmapCommit>$}{$\Delta,c,t$} \<return> \( \Delta[c].\<failmap>[t] = \nofail \) 
     \EndFunction    
     \Function{$\<failmapFail>$}{$\Delta,c,t$}
        \<return> \( \Delta[c].\<failmap>[t] = \fail \) 
     \EndFunction   
     \Function{$\<timeoutCommit>$}{$\Delta,c,t$}
        \<return> \( \Delta[c].\<timeout>[t] = \nofail \) 
     \EndFunction
     \Function{$\<undecided>$}{$\Delta,c,t$}
          \<return> \( \Delta[c].\<failmap>[t] = \undefined \) 
     \EndFunction   
     \Function{ $ \<monitoringContracts> $ }{$ l,t $}
          \<return> \( \{ c : l.\Delta[c].\<failmap>[t] \neq \None \} \)
     \EndFunction   
     \Function{ $ \<allMonitoringCommit> $ }{$ l,t $}
        \State \<return> \( \forall c \in \<monitoringContracts>(l,t): \<failmapCommit>(l.\Delta,c,t) \)
     \EndFunction
     \Function{ $ \<oneMonitoringFail> $ }{$ l,t $}
        \State \<return> \( \exists c \in \<monitoringContracts>(l,t): \<failmapFail>(l.\Delta,c,t) \)
     \EndFunction 
     \Function{$\<commitWithTimeout>$}{$\Delta,c,t$}
        \State \<return> \( \<failmapCommit>(\Delta,c,t) \vee
        (\<undecided>(\Delta,c,t) \wedge \<timeoutCommit>(\Delta,c,t)
        \)
     \EndFunction   
     \Function{ $ \<allMonitoringCommitWithTimeout> $ }{$ l,t $}
        \State \<return> \( \forall c \in \<monitoringContracts>(l,t): \<commitWithTimeout>(l.\Delta,c,t) \)
     \EndFunction
    \end{algorithmic}
    \vspace{-1em}
  }
\newcommand{\FigAux}{
  \begin{figure}[h!] 
    \AuxBody
    \caption{Auxiliary functions}
    \label{fig:alg:aux}
  \end{figure}
}


Let \((H,\tau)\) be a blockchain run and \(t\) a transaction.
We extend the monitoring tree \(\tau\) by adding a new level attaching
\(t\) from every possible leaf, which increases by one the height of
\(\tau\) (see Fig.~\ref{fig:alg:stepextend}).
Let \(\tau'\) be the result of \(\<extend>(\tau,t)\).
If \(\tau'\) has height \(k+1\), the monitoring window for the first
transaction in $\tau'$ has expired and its monitor must fail or
commit.
To take this decision, function \(\<step>\) invokes function $\<prune>$.
The resulting monitoring tree $\tau''$ returned by function $\<prune>$
becomes the new monitoring tree.
Finally, function \(\<step>\) extends $H$ making the first pending
transaction permanent.

\begin{figure}[t!]
  \begin{tabular}{|@{\!\!}l@{\!\!}|@{\!\!}l|}\hline
    \begin{minipage}{0.36\textwidth}
      \vspace{0.9em}
      \StepBody
      \vspace{0.2em}
    \end{minipage}
    &
      \begin{minipage}{0.66\textwidth}
        \ExtendBody
      \end{minipage}\\ \hline
  \end{tabular}
  \caption{Functions \(\protect\<step>\) and \(\protect\<extend>\).}
  \label{fig:alg:stepextend}
\end{figure} 

Function $\<prune>$ (see Fig.~\ref{fig:alg:funcs}) determines whether
to commit or fail the first pending transaction \(tx\) in monitoring
tree $\tau$ with height \(k+1\) returning either the committing
or failing subtree of $\tau$.
If \(\tau\) has only one successor, the decision is trivial, otherwise
we analyze \(tx\) possible futures.
Function \(\<prune>\) checks all futures assuming \(tx\) commits,
(i.e., all leaves in the committing subtree of \(\tau\)); if the
future monitor of transaction \(tx\) commits in all of them, then
\(tx\) commits and the committing subtree of \(\tau\) becomes the new
monitoring tree.
Otherwise, \(tx\) fails and the failing subtree of \(\tau\) becomes
the new monitoring tree.
If \(\<prune>\) cannot assert whether the monitored transaction fails
or commits, \(\<prune>\) invokes \(\timeout\) to decide (see
function~\(\<allMonitoringCommitWithTimeout>\) in
Fig.~\ref{fig:alg:funcs}).



\begin{figure}[t!]
  \begin{tabular}{@{}|@{}l@{}|@{}}\hline
    \begin{minipage}{\textwidth}
      \vspace{0.5em}
      \PruneBody
      \vspace{1em}
    \end{minipage}  \\[0.2em] \hline
    \begin{minipage}{\textwidth}
      \vspace{0.5em}
      \InnerPruneBody
      \vspace{1.5em}
    \end{minipage}    \\[0.2em] \hline
    \begin{minipage}{0.95\textwidth}
      \vspace{0.5em}
      \AuxBody
      \vspace{1.5em}
    \end{minipage} \\[0.2em] \hline
  \end{tabular}
  \caption{Functions $\mathsf{decide}$, $\mathsf{prune}$ and auxiliary functions.}
  \vspace{-1em}
  \label{fig:alg:funcs}
\end{figure}


In some cases, the decision of future monitors is known before the
monitoring windows ends.
In such instances, some nodes are unreachable, called
\emph{impossible nodes}.
For example, when a transaction future monitor is waiting for a
transaction in the future and that transaction happens before the
monitoring window ends, the future monitor is going to be set to
commit, which turns all nodes in its failing subtree impossible nodes.
Concretely, if in all possible futures in the committing subtree of
node $n$ its transaction is known to commit, then all nodes in the
failing subtree of $n$ are impossible nodes.
Similarly, if in all possible futures in the committing subtree of
node $n$ its transaction is known to fail, then all nodes in the
committing subtree of $n$ are impossible.
Impossible nodes are removed before deciding whether a transaction
commits or not, since we may incorrectly deduce that a monitor fails
because of an impossible future node.
Consequently, $\<prune>$ invokes $\<innerprune>$ to remove all
impossible nodes, and only then, \(\<prune>\) determines whether the
root transaction commits or not as explained above.
Function $\<innerprune>$ (see Fig.~\ref{fig:alg:funcs}) shows how
to prune impossible nodes from trees.
To guarantee that impossible nodes are pruned before checking if roots
of trees are impossible (either commit or fail), we perform a
bottom-up recursion.


\newcommand{\stepfig}[2]{\includegraphics[width=#2\textwidth]{figures/step-#1.pdf}}

\begin{figure}[t!]
  \footnotesize\centering
  \noindent
  \begin{tabular}{@{}c@{}c@{}}
  \stepfig{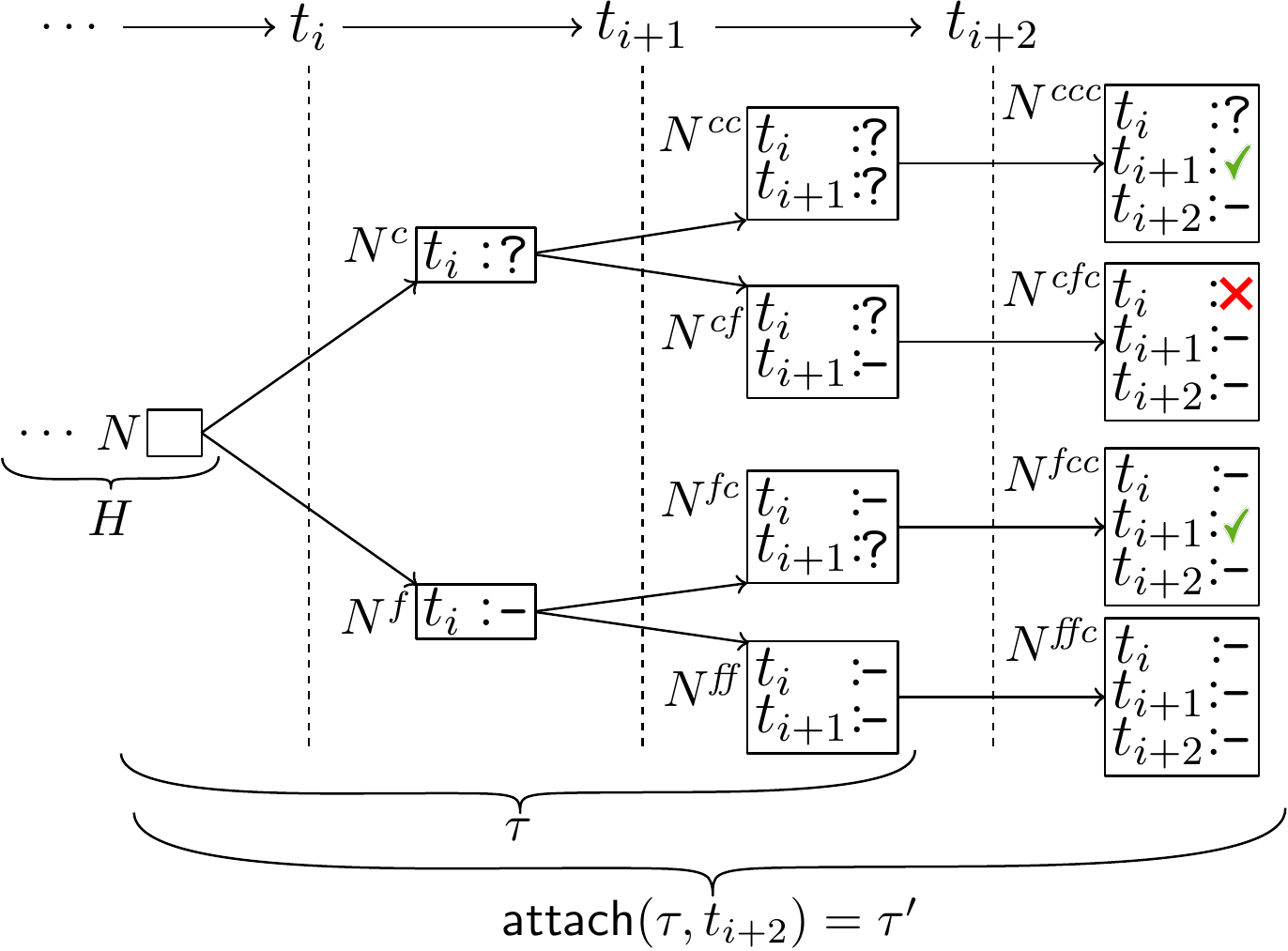}{0.50} & \stepfig{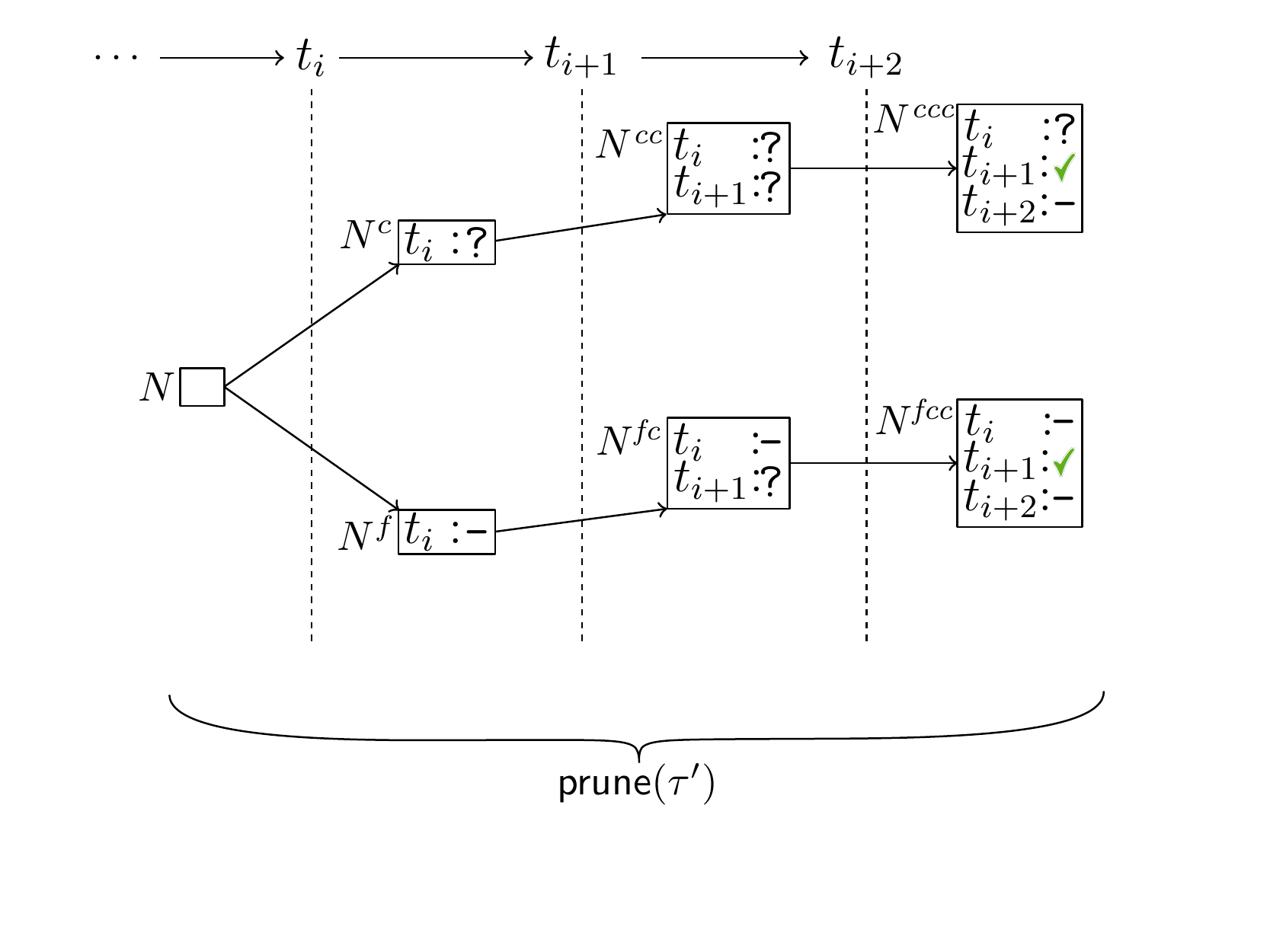}{0.50} \\
      (a)  & (b) \\
  \end{tabular}
  \stepfig{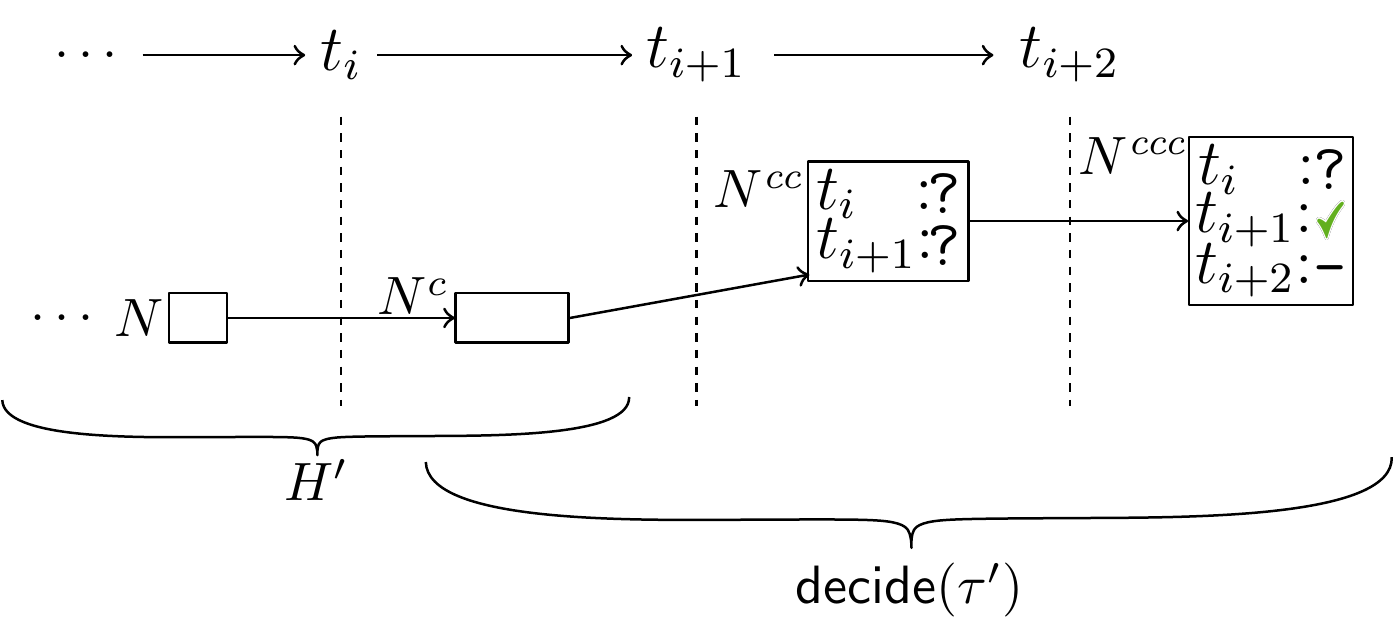}{0.5} \\
  (c)
  \caption{Application of function \(\protect\<step>\) in a blockchain run.}
  \label{fig:step}
\end{figure}

\begin{example}
  Fig.~\ref{fig:step} shows the result of applying function
  \(\<step>\) to blockchain run \((H,\tau)\) with a monitoring window
  \(k=2\) and two pending transactions \(t_i\) and \(t_{i+1}\).
Each node in the monitoring tree is annotated with the monitor state
of all pending transactions up to that node: a question mark means
undecided monitors, a tick means known to commit monitors, a cross
means known to fail monitors, and a dash denotes no monitored
transactions.
Initially, no monitors are decided in any node in \(\tau\).

Function \(\<step>((H,\tau),t_{i+2})\) first invokes function
\(\<extend>(\tau, t_{i+2})\).
This function adds a new level to \(\tau\) by applying transaction
\(t_{i+2}\) at all leaves in \(\tau\), obtaining monitoring tree
\(\tau'\), Fig.~\ref{fig:step}(a).
Transaction \(t_{i+2}\) immediately commits at all leaves in \(\tau\),
generating nodes \(N^\ccc, N^\cfc,N^\fcc\) and \(N^\ffc\).
%
%
The future monitor for transaction \(t_{i}\) is known
to fail at node \(N^\cfc\) while remaining undecided at node
\(N^\ccc\) and the future monitor for transaction \(t_{i+1}\) is
known to commit at nodes \(N^\ccc\) and \(N^\fcc\).
Next, as the height of the new monitoring tree, \(\tau'\), is \(3>2\),
function \(\<step>\) invokes function \(\<prune>(\tau')\) to decide if
the first pending transaction, \(t_{i}\), fails or commits.
Function \(\<prune>\) invokes function \(\<innerprune>\) to
remove all impossible nodes in \(\tau'\).
%
%
When computing \(\<innerprune>\), the failing subtree of node \(N^c\),
rooted at node \(N^\cf\), is removed because at node \(N^\ccc\) the
future monitor for the transaction at node \(N^c\), \(t_{i+1}\), is
known to commit and node \(N^\ccc\) is the only future in the
committing subtree of node \(N^c\), making the subtree rooted at
\(N^\cf\) an impossible subtree.
Similarly, the subtree rooted at \(N^\ff\) is an impossible subtree
and it is also removed by function \(\<innerprune>\).

Subtrees with roots \(N^\cf\) and \(N^\ff\) are the only ones removed
when applying function \(\<innerprune>\) to monitoring tree \(\tau'\),
as shown in Fig.~\ref{fig:step}(b).

Finally, to decide whether to commit or not transaction \(t_i\)
function \(\<prune>\) consider node \(N^\ccc\), as it is the only
future in the committing subtree of node \(N\) in the monitoring tree
returned by function \(\<innerprune>\).
At node \(N^\ccc\) the future monitor for transaction \(t_{i}\) is
undecided.
However, since its monitoring window has ended, function
\(\<prune>\) uses the \(\<timeout>\) of the contracts that are
undecided.
Assuming for all undecided contracts their \(\<timeout>\) function commit
transaction \(t_i\), then function \(\<prune>\) commits transaction \(t_i\),
returning the subtree rooted at \(N^c\) as the new monitoring tree (see
Fig.~\ref{fig:step}(c)), it would fail if at least one contract timeout function
fails.
Finally, function \(\<step>\) extends \(H\) by making transaction
\(t_i\) permanent.
If \(\<innerprune>\) had not been applied before function \(\<prune>\)
evaluated all futures in the committing subtree of \(N\), transaction
\(t_{i}\) would have incorrectly failed, as in impossible future
\(N^{cfc}\), the future monitor for transaction \(t_{i}\) fails.
\end{example}

Appendix~\ref{app:TwoBounded} shows an example of contracts that only
lend their tokens if they receive them back within $2$ transactions in
the future.


\section{Properties}\label{sec:properties}

We discuss now properties of the model of computation defined in
Section~\ref{sec:modelOfComputation}.
In particular, we establish how the new model extends the previous
one, that the size of monitoring trees is manageable,
and the blockchain always progresses.
We assume a fixed monitoring window $k$.
All proofs can be found in Appendix~\ref{sec:proofs}.


After the monitoring window has expired, the root transaction is
confirmed and one of two possible successors is consolidated.
\begin{replemma}{monitoring-tree}\label{l:monitoring-tree}
  Let \((H,\tau)\) be the system run after \(k\) transactions, \(t\) a
  transaction and \((H',\tau') = \<step>((H,\tau),t)\).
%
The root of $\tau'$ is one of the successors of the root of $\tau$ and
all paths in $\tau'$ without leaves are also paths in $\tau$.
Moreover, $H'$ is obtained by extending $H$ with the first pending
transaction on $\tau$.
\end{replemma}

The first $k$ transactions from the genesis are just added to the
tree.
From the previous lemma, after \(k\) transactions and when
a new step is taken, the first pending transaction is either committed
or failed and a new pending transaction is attached to all leaves.
Moreover, the transaction added to the history is the root of the
previous monitoring tree and one of its successors is the root of the
new monitoring tree.
In other words, exactly one of the paths in the monitoring tree
eventually becomes permanent, and thus, the blockchain always
progresses.
\begin{corollary}[Progress]~\label{coro:progress}
  Function $\<step>$ is total and, after the first \(k\) invocations,
  each execution of $\<step>$ makes one transaction permanent.
\end{corollary}



The height of the monitoring tree is bounded by the monitoring window.

\begin{replemma}{bounded-certainty}[Bounded Certainty]\label{l:bounded-certainty}
  Let \(\tau\) be a monitoring tree in a blockchain run obtained by
  applying function step \(l\) times.
  Then, the height of \(\tau\) is the minimum between \(l\) and \(k\).
  Moreover, all leaves in \(\tau\) are in its last level.
  \end{replemma}
Function $\<innerprune>$ removes all impossible nodes from monitoring trees.
%
Function $\<innerprune>$ recursively removes impossible nodes in the committing
and failing subtrees, and then, determines if it can remove any subtree by
inspecting all possible futures in the committing successor.
\begin{replemma}{innerprune-subtree}\label{l:innerprune-subtree}
  Function $\<innerprune>(\tau)$ returns a sub-monitoring tree of $\tau$ without
  impossible nodes and only impossible nodes were removed.
  %
  \end{replemma}

Function \(\<step>\) consistently makes the blockchain progress.
%
After more than $k$ transactions were added, the first pending transaction
is made permanent (see Corollary~\ref{coro:progress}).
The resulting monitoring tree keeps the order of the rest of the pending
transactions and it also preserves the same information of the pending
transactions except the last.

\begin{replemma}{prune}\label{l:prune}
  Let \(\tau\) be a monitoring tree, $\eta$ be the result of expanding
  \(\tau\) with a new transaction, \(t\) be the first pending
  transaction in \(\tau\), and \(\nu\) be the decided subtree of
  \(\eta\).
  \begin{compactitem}
  \item If \(\eta\) has only one successor then \(\nu\) is the
    result of pruning \(\eta\)'s successor.
  \item If \(\eta\) has two successors, then let \(\eta_{c}\) and
    \(\eta_{f}\) be the result of pruning the committing and failing
    subtrees of \(\eta\) respectively.
    \begin{compactitem}
    \item Monitoring tree \(\nu\) is \(\eta_{c}\) if in all possible futures
    assuming \(t\) commits, transaction \(t\) does not fail or if no
      decision has been reached, all pending $\<timeout>$ functions of $t$
      commit.
    \item Monitoring tree $\nu$ is $\eta_f$ if there is a possible future where
      assuming transactions \(t\) commits, leads to the monitor of $t$
      fail or some of the pending $\<timeout>$ function of $t$ fail.
  \end{compactitem}    
\end{compactitem}
\end{replemma}

The size of monitoring trees can be exponential in the number of
monitored transaction rather than in the monitoring window size, as
monitored transactions are the only ones branching monitoring trees.
\begin{replemma}{size}\label{l:size}
  Let \(\tau\) be a monitoring tree and \(m\) be the number of
  monitored transactions in \(\tau\) (so \(m \leq k\)).
Then, the size of \(\tau\) is in \(\mathcal{O}(2^m \times k)\).
\end{replemma}

In practical scenarios, the number of monitored transactions typically
is small compared to the monitoring window because most transactions
do not require future monitors.
This makes the size of the monitoring tree much smaller than the
theoretical maximum.

\begin{corollary}\label{coro:size}
  If the number of monitored transactions in monitoring trees is
  constant then the size of monitoring trees is bounded by
  \(\mathcal{O}(k)\).
\end{corollary}


Finally, we show that adding future bounded monitors preserves legacy
executions, so for blockchain runs where no contracts use future
monitors, the monitoring tree is a chain with no branching.

A legacy monitoring tree \(\tau\) is such that every configuration
obtained from applying $\applyTx$ coincides with rule \(\leadsto\).

\begin{replemma}{legacy-tx}[Legacy Pending Transactions]~\label{l:legacy-tx}
  Let \(\tau\) be a legacy monitoring tree.
  Then, $\tau$ is a chain and the effect of executing all transactions
  in \(\tau\) is equivalent to executing them in the traditional model
  of computation.
\end{replemma}

If we add that the permanent history is equivalent (up to now)
to the traditional model, then the evolution of the blockchain in both
models coincide.

\begin{replemma}{legacy}[Legacy History]~\label{lemm:legacy}
  Let \(\tau\) be a legacy monitoring tree and \(H\) be a history
  such that every permanent transaction coincides with rule \(\leadsto\).
  Then, the result of concatenating \(H\) and \(\tau\) is equivalent
  to the traditional model of computation.
\end{replemma}

From Corollary~\ref{coro:progress} and Lemma~\ref{lemm:legacy}, we
conclude that the new model of computation is consistent with the
previous model of computation and eventually creates a chain.
Additionally, Corollary~\ref{coro:size} implies that in practical
scenarios, the size of monitoring trees is linear on the monitoring
window, making it a feasible and practical blockchain implementation.


\section{Atomic Loans}\label{sec:example}

Flash loan contracts allow other contracts to borrow tokens
\emph{without any collateral} only if the borrowed tokens are repaid during
the same transaction~\cite{EIP.FlashLoan} (typically with some
interest).
Atomic loans are a generalization of flash loans where the borrowing
party can repay the lending party in future transactions.
It is not possible to implement flash loans unless additional mechanisms are
added to the blockchain~\cite{capretto22monitoring}.
Similarly, it is impossible to implement atomic loans in traditional
blockchain computational models.
As transaction monitors~\cite{capretto22monitoring} enable flash loans
transactions, future monitors allow monitors to check properties
across transactions enabling atomic loans.
We illustrate now how to implement atomic loans using the monitoring
window as the maximum payback time.

\newcommand{\FLSafety}{\textbf{AL-safety}\xspace}
\newcommand{\FLProgress}{\textbf{AL-progress}\xspace}
We specify lender contracts as contracts respecting the following two
properties:
\begin{specification}[Atomic Loans]\label{spec:atomic-loans}
  We say contract \(\mathsf{A}\) is an atomic lender if:
\begin{compactdesc}
\item[AL-safety:] A loan from \(\mathsf{A}\) is repaid to \(\mathsf{A}\) within the monitoring window.
\item[AL-progress:] Contract $\mathsf{A}$ grants loans unless
\textup{\FLSafety} is violated.
  \end{compactdesc}
  \end{specification}

The following contract \lstinline|FlashLoanLender| shows a simple contract
implementing a flash loan lender\footnote{Flash loan lender are atomic loan
lenders with paying back window of one.} using \fnfhookup
hookup~\cite{capretto22monitoring}, i.e. with no future monitors but transaction
monitors.
%
%
We highlight monitor code with gray background.

\vspace{1em}
\noindent\begin{tabularx}{\textwidth}{|@{}X@{}|}\hline
  
  \includegraphics[scale=0.53]{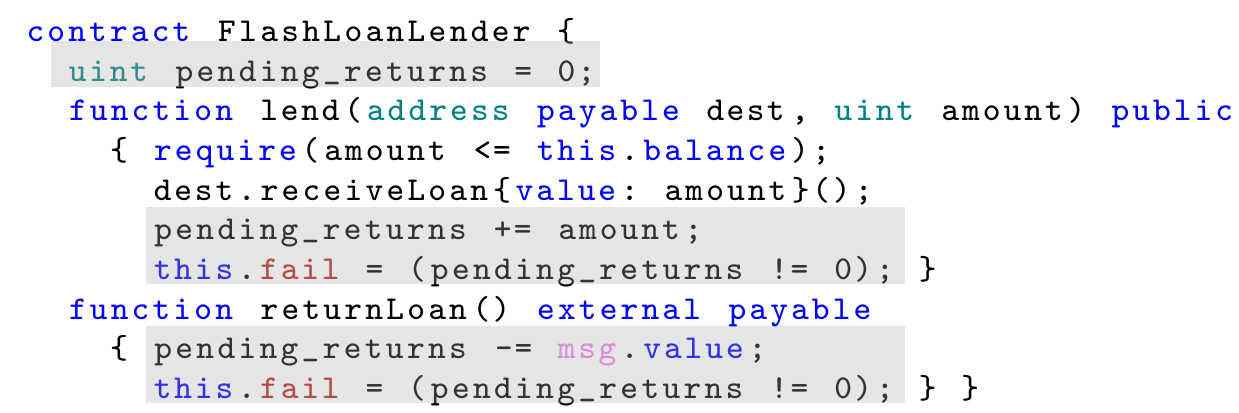} \\ \hline
\end{tabularx}
\vspace{1em}

Function \lstinline|lend| lends as long as the lender has
enough funds, annotates the borrowed tokens in
\lstinline|pending_returns| and sets its \lstinline|fail| bit so the
transaction commits only if the loan is paid back.
When the loan is returned, \lstinline|returnLoan| decreases
\lstinline|pending_returns| and updates its \lstinline|fail| bit.
At the end of each transaction, if there are pending loans the
\lstinline|fail| bit will make the transaction fail.

The above contract implements flash loans that must be returned within a
transaction, but does not work properly if future transactions are considered.
It is not possible to successfully predict or check whether the loan
is returned in some future transactions.
We show now how future monitors solve this problem.

The following contract \lstinline|Lender| is an atomic lender
using future monitors.
%
All loans are treated equally and should be paid back on time, and if
one loan is not returned, then all loans issued at the same transaction
would be rejected.
Here we are being too strict compared to practical cases, but it is
enough to illustrate the use of future transaction monitors.

\vspace{1em}
\noindent\begin{tabularx}{\textwidth}{@{}|X@{}|}\hline
  \includegraphics[scale=0.53]{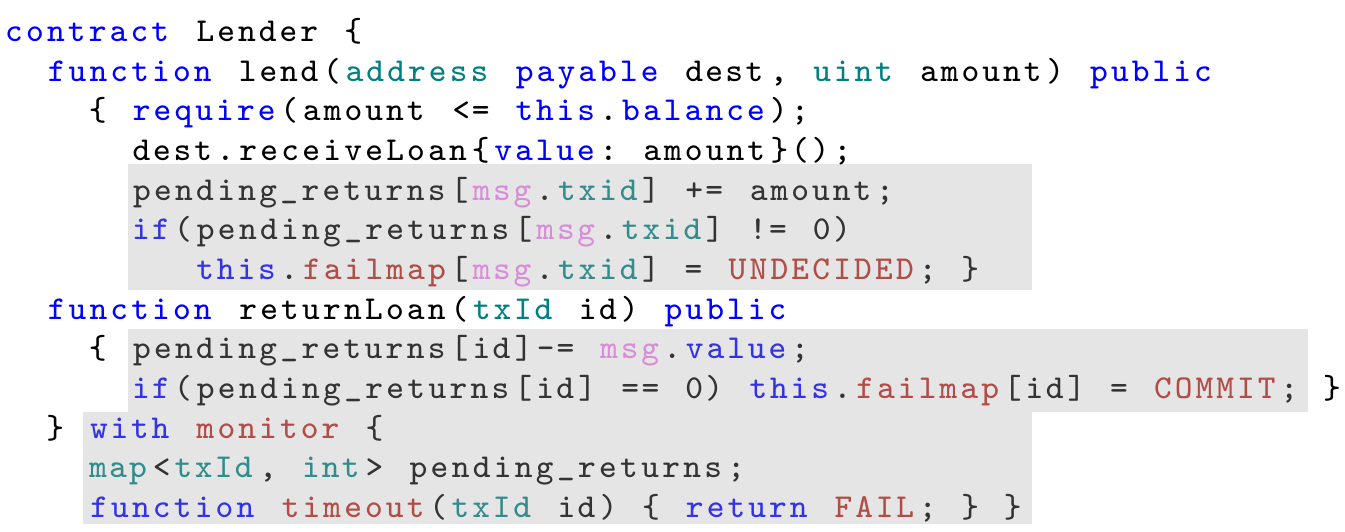} \\ \hline
\end{tabularx}
\vspace{1em}

\noindent Contract \lstinline|Lender| uses a map
\lstinline|pending_returns|, from transactions to the amount borrowed
within that transaction, to determine whether a transaction should
commit or fail.
Function \lstinline|lend| grants a loan if the lender has enough
funds, increases the corresponding entry in map
\lstinline|pending_returns| for the current transaction and sets the
\lstinline|failmap| entry activating the current transaction monitor.
Client contracts can repay loans by invoking \lstinline|returnLoan|,
which receives the transaction identifier of the lending transaction
to decrease the corresponding entry in \lstinline|pending_returns| by
the amount received.
If \lstinline|pending_returns| reaches 0 for a given transaction, the
\lstinline|failmap| entry of that transaction is set to
\lstinline|COMMIT|.
Finally, \lstinline|timeout| returns \lstinline|FAIL| to fail
transactions with unpaid loans at the end of their monitoring window.

Clients can request loans without further collateral, satisfying
\FLProgress, and if loans are not returned within the monitoring
window, the lending transaction will retroactively fail, satisfying
\FLSafety.

The following contract \lstinline|NaiveClient| requests a loan
invoking \lstinline|borrow|.

\vspace{1em}
\noindent\begin{tabularx}{\textwidth}{|X|}\hline
  \includegraphics[scale=0.53]{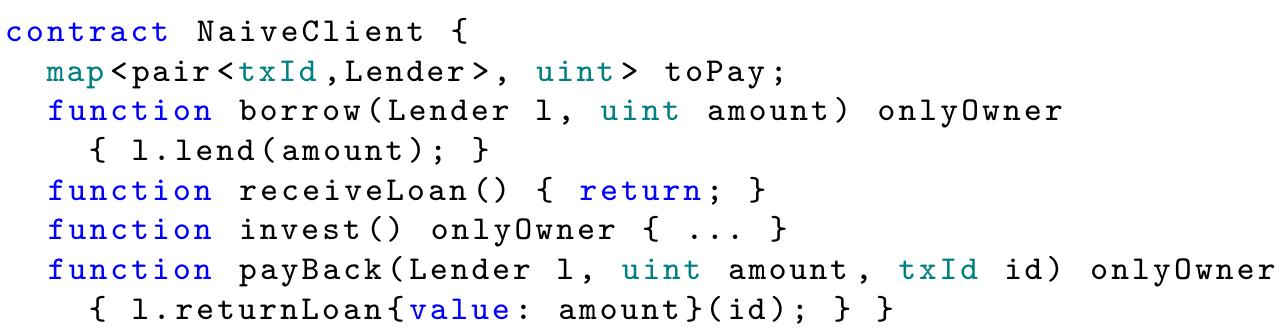} \\ \hline
\end{tabularx}
\vspace{1em}

In subsequent transactions, the client can invest the funds, and in a
final transaction, return the loan to the lender invoking
\lstinline|payBack|.

\newcommand{\NC}{\texttt{NC}\xspace}

Let \NC and \texttt{L} be two contracts installed in a
blockchain with a monitoring window of length \(2\), where \NC
runs \lstinline|NaiveClient| and \texttt{L} runs \lstinline|Lender|.
Consider \((\Sigma,\Delta)\) to be the current state of the blockchain at which
\texttt{NC} has 100 tokens and \texttt{L} has 1000 tokens.
From \((\Sigma,\Delta)\), the sequence of transactions is:
\begin{inparaenum}[(1)]
\item \NC requests a loan,
\item \NC invests assuming contract \texttt{L} lends the money, and
\item \NC returns the loan.
\end{inparaenum}
Because \texttt{L} employs future monitors to guarantee clients pay
back, the first transaction generates a branching on the blockchain
evolution.
The next two transactions are not monitored, thus they do not create
any branching.
Therefore, after these three transactions, there exist two possible
futures as shown in Fig.\ref{fig:lender_client_alwayspays:simp}, one
where \texttt{L} grants the loan and another where it does not.
\begin{figure}[t!]
  \centering
  \includegraphics[scale=0.4]{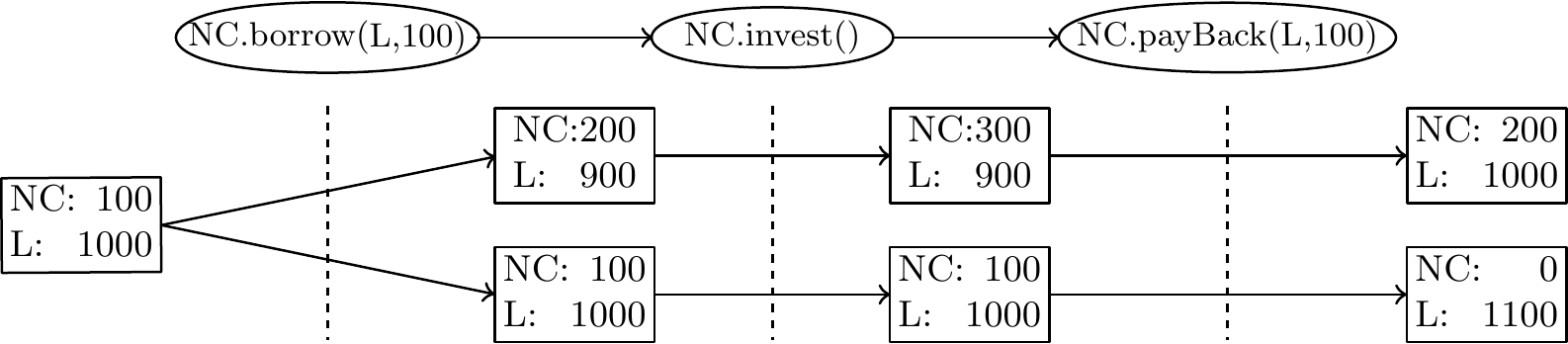}
  \caption{Balance of contracts \texttt{NC} and \texttt{L} in the
    monitoring tree after executing the three transactions posted by a client.}
  \label{fig:lender_client_alwayspays:simp}
\end{figure}
We can see that \NC pays back in all possible futures.
Moreover, contract \NC pays back even in the future where
contract \texttt{L} fails the past lending operation
(for a detailed
 explanation see Appendix~\ref{app:Example}).

A malicious lender can take advantage of such behavior, for example
using the following contract \lstinline|MaliciousLender|.

\vspace{1em}
\noindent\begin{tabularx}{\textwidth}{|X|}\hline
  \includegraphics[scale=0.53]{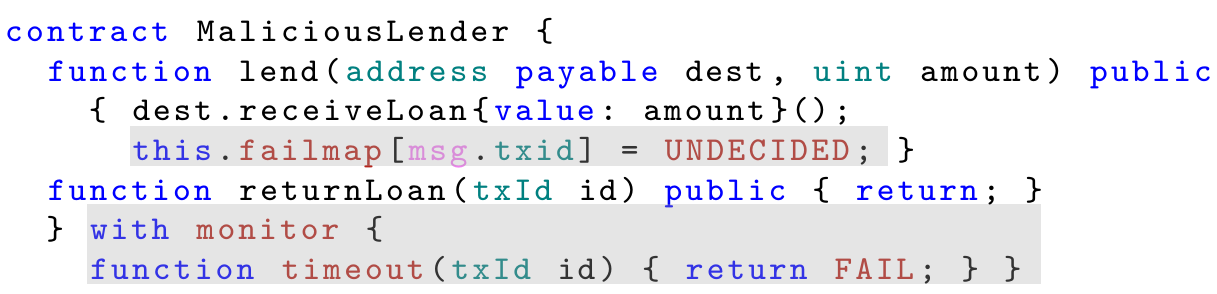} \\ \hline
\end{tabularx}
\vspace{1em}

The above malicious lender, upon receiving a loan request in
function \lstinline|lend|, if it has enough tokens, it grants the loan
and marks the transaction as undecided using its \lstinline|failmap| map.
However, this lender contract does not update its \lstinline|failmap| map
when receiving paybacks.
Therefore, at the end of the monitoring window, the monitor remains
undecided making the lending transaction fail due to the
\lstinline|timeout| function.
In other words, the malicious lender never lends any tokens,
as all its loans are reverted, but it looks like it does.
When combined with \lstinline|NaiveClient| and the same three
transactions described earlier, the malicious lender will receive the
repayment of a loan from client \NC without having given the
loan.
In Fig.~\ref{fig:lender_client_alwayspays:simp}, the bottom branch is the one
that survives when the lender implements a malicious contract.

The problem arises because client \NC does not implement any
mechanism to check in which branch it is executing when repaying the
loan.
The naive contract does not distinguish between the scenario where the
loan will ultimately be committed and the scenario where it will fail.
As a result, client \NC ends up providing payments in both
cases.

The following contract \lstinline|Client| presents a correct client
implementing a map, \lstinline|toPay|, to keep track of its debts to
lenders.

\vspace{1em}
\noindent\begin{tabularx}{\textwidth}{|X|}\hline
  \includegraphics[scale=0.53]{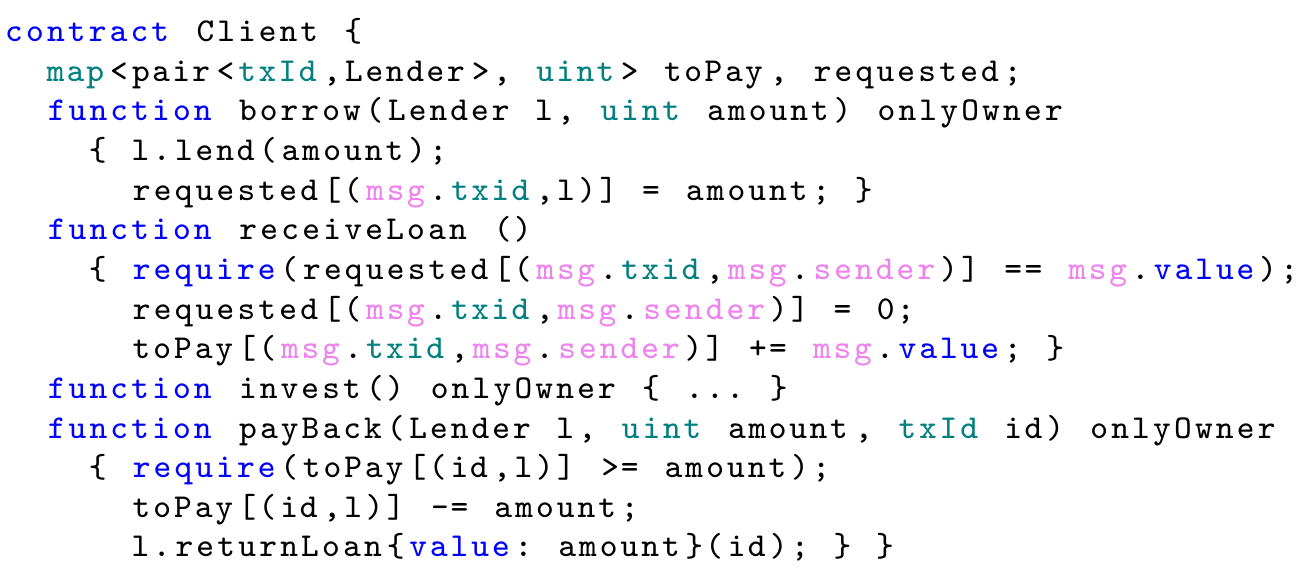} \\ \hline
\end{tabularx}
\vspace{1em}

The above contract allows clients to determine the specific path in
which it is executing, and thus, to decide whether to repay.
Consequently, clients can successfully get loans from correct
lenders while being resistant to attacks from malicious lenders.

Fig.~\ref{fig:lender_client} shows an execution following the same transactions
as before but with the correct contract \lstinline|Client|: clients request
a loan, invest the money, and payback the loan.
The top branch shows the case where the lender sends the money and the
client returns it, while the bottom branch shows the case where the
loan is not given.
In the former cases, the client returns the money, and in
the latter case, the client just fails the transaction.
\begin{figure}[t!]
     \centering
     \includegraphics[scale=0.4]{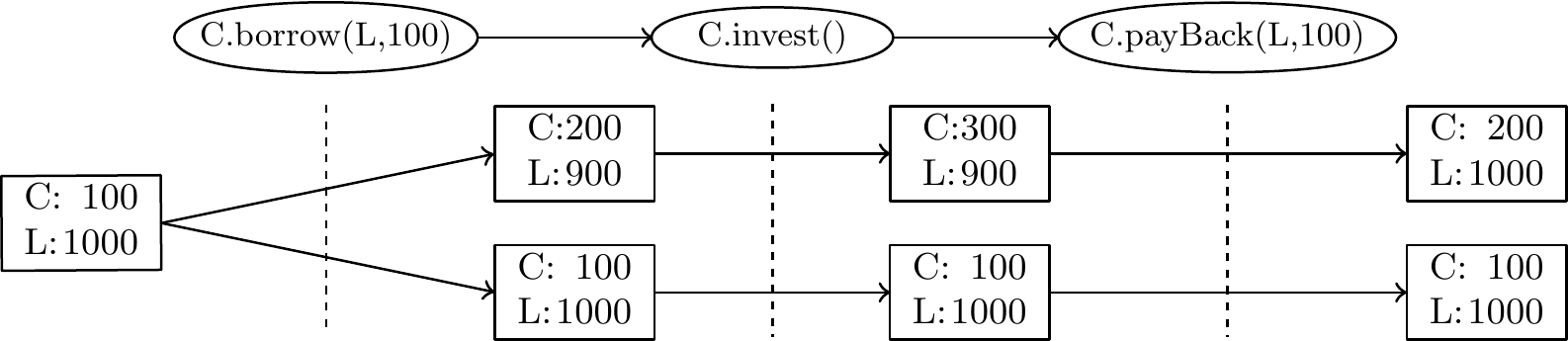}
     \caption{Balance of contracts \texttt{C} and \texttt{L} in the
monitoring tree after executing the three transactions posted by a client.}
    \label{fig:lender_client}
\end{figure} 

These examples show how even contracts not monitoring transactions
need to be aware that transactions can create potential executions in
the blockchain evolution that may be reverted due to future monitors.
Since the same transaction is executed in all possible scenarios, but
their effects may be different, contracts need to know in which
temporal line they are executing and act accordingly.
Contract \lstinline|Client| accomplishes this by maintaining a record
of debts owed to lenders in variable \lstinline|toPay|.


\section{Related Work}\label{sec:related-work}

\myparagraph{Dynamic verification of smart contracts}
Runtime monitoring tools like
ContractLarva~\cite{ellul18runtime,azzopardi18monitoring} and
Solythesis~\cite{li20securing} take a smart contract code and its
properties as input and produce a safe smart contract that fail
transactions violating the given properties.
They achieve this be injecting the monitor into the smart contract as
additional instructions.
Therefore, these monitors are restricted to one operation in a single
contract.
Transaction Monitors~\cite{capretto22monitoring} extend monitoring
beyond a single operation to observe the effect of an entire
transaction execution on a given contract.

While these existing works provide strong foundations for smart
contract verification, none directly address the ability to react
based on future transactions, as proposed in this work.

       
\myparagraph{Branching Computational Models} The monitoring tree
generated by pending transactions might reassemble the tree-like
structure in branching-time logic such as CTL~\cite{clarke1982design}.
However, it is worth noting that the branching in the monitoring tree
represent all possible futures given by the monitors of the pending
transaction, and exactly one path eventually consolidates.
More important, future monitors are not aware of the existence of the
other paths in the monitoring tree and therefore cannot reason about
them.
CTL, on the other hand, can be used to express properties that reason
about the different paths in the tree.


\section{Conclusion}\label{sec:conclusion}
We presented future monitors for smart contracts.
Future monitors are a defense mechanism enabling contracts to state
properties across multiple transactions.
These kinds of properties are motivated by long-lived transactions, in
particular by atomic loans, which are not implementable in their full
generality in current blockchains.
To implement future monitors, we introduced the notion of monitoring window and
two additional new mechanisms to blockchains, namely failing maps and timeout
functions.

Future monitors delay the consolidation of transactions, but the
system remains consistent and we gain in expressivity.
The outcome of transactions remains deterministic and depends solely
on the transactions themselves, but now transactions can fail because
of future actions.
Combining all elements we obtained a deterministic semantics with
future monitors in place.

We have also illustrated that contracts need to be aware of the
existence of possible executions.
Future monitors introduce a branching model to describe the evolution
of blockchain systems where transactions may commit or not, caused by
the temporary uncertainty regarding the effect of pending
transactions.
Consequently, when new transactions are added to the blockchain, they
are executed in multiple blockchain configurations, representing
possible time-lines.
Therefore, contracts need to be aware of the different contexts in
which they are executing, ensuring that the transaction produces the
desired effects in all possible realities.

The main contribution of this paper is theoretical and we left the
full implementation of future monitors as future work.
Optimistic rollup systems, where the effect of transactions is already
delayed due to the fraud-prove arbitration scheme, present an ideal
environment to incorporate future monitors into practical blockchain
systems without further implications.
In particular, optimistic rollup systems can allow future transaction
monitors with little modifications, and more importantly, without
modifying the underlying blockchain.

For simplicity, we have neglected a specific analysis of the
additional gas consumption that arises for using future monitors,
which might lead to the failure of accepting transactions.
Nevertheless, we conjecture that future monitors are simple enough to
guarantee that a calculable amount of gas will prevent gas failing
situations.
However, we leave a detailed study for future work.



 \bibliographystyle{abbrv}
 \bibliography{bibfile}

 \vfill\pagebreak
 
\appendix
\section{Two-bounded Monitor Example}
\label{app:TwoBounded}

%
In this section we show an example of the evolution of the system
proposed in Section~\ref{sec:blockevo}.
The example consists of three contracts exchanging two tokens between
them in a blockchain accepting a monitoring window of length \(2\).

Let \(a,b,c\) be three contracts installed in the blockchain such that
both \(a\) and \(b\) have one token each.
Contracts $a$ and $b$ use future monitors and only send their tokens
if they will receive their (respective) token back within 2
transactions.
For simplicity, in this example, we neglect gas consumption.

Let \((\Sigma_{0},\Delta_0)\) be a blockchain configuration such that
contracts \(a,b,c\) are installed and \(a\) and \(b\) have a token
each while \(c\) has none.
In Fig.~\ref{fig:traces}, we show the evolution of the system
beginning from configuration \((\Sigma_{0},\Delta_0)\) and executing the
following three transactions in order:
\begin{inparaenum}[(1)]
\item transaction \(t_{a,c}\), smart contract \(a\) sends its token to
\(c\);
\item transaction \(t_{b,c}\),
smart contract \(b\) sends its token to \(c\);
\item transaction \(t_{c, \langle a, b \rangle}\), smart contract \(c\)
sends tokens to \(a\) and \(b\).
\end{inparaenum}
The transaction \(t_{c,\langle a, b \rangle}\) involves several
internal operations where \(c\) sends a token to \(a\) and \(b\).
In Fig.~\ref{fig:traces}(a), we show the initial blockchain run
$(H_0,\tau_0)$ where both \(H_{0}\) and \(\tau_{0}\) have a node
\(N_0\) with the initial configuration \((\Sigma_{0},\Delta_{0})\).
Let \((H_{1}, \tau_{1})\) be the resulting run after contract \(a\)
sends its token to contract \(c\), i.e. executing transaction
\(t_{a,c}\)~(See
Fig.~\ref{fig:traces}(b-\(\<step>((H_0,\tau_0),t_{a,c})\))).
Blockchain run \((H_{1}, \tau_{1})\) is the result of executing
transaction \(t_{a,c}\) to history \(H_{0}\) and monitoring
tree \(\tau_{0}\), in symbols,
\(\<step>((H_{0},\tau_{0}),t_{a,c}) = (H_{1}, \tau_{1})\).
Internally, function \(\<step>\) firsts extends $\tau_0$ by executing
transaction $t_{a,b}$ in its only leaf, $N_0$.
Since $a$ only sends its token if it can get it back, the result of
executing $t_{a,c}$ in $N_0$ depends on future transactions creating a
branching in the monitoring tree.
In its committing successor, $N^c_1$, contract \(a\) token is sent to
$c$, as the transaction is assumed to commit in this scenario.
In the failing successor, $N^f_1$, contract $a$ keeps its token.
Since the monitoring tree has height $1 < 2+1$, the monitor window for
the first pending transaction has not expired, and the history remains
unchanged, $H_1 = H_0$.

Similar to the first transaction, now we execute the second
transaction where contract \(b\) sends its token to contract \(c\)
resulting in another run $(H_2,\tau_2)$~(See
Fig.\ref{fig:traces}(c-\(\<step>((H_1,\tau_1),t_{b,c})\))).
First, it extends $\tau_1$ by executing transaction $t_{b,c}$ in its
leaves, $N^c_1$ and $N^f_1$.
In both configurations, contract $b$ has a token and only sends it if
it can get it back within the monitoring window, and thus, we branch
again in both cases generating four possible configurations:
\begin{itemize}
\item At configuration $N^{\cc}_2$, the committing successor of
  $N^{c}_1$, contract \(c\) has both tokens, as transactions $t_{a,c}$
  and $t_{b,c}$ are assumed to commit in this scenario.
\item At configuration $N^{\cf}_2$, the failing successor of
  $N^{c}_1$, contract \(b\) has its token and contract \(c\) has
  contract \(a\) token, as transaction $t_{a,c}$ is assumed to commit
  while transaction $t_{b,c}$ is assumed to fail.
\item At configuration $N^{\fc}_2$, the committing successor of
  $N^{f}_1$, is the dual case of the previous one.
\item At configuration $N^{\ff}_2$, the failing successor of
  $N^{f}_1$, contracts \(a\) and \(b\) have their tokens, as
  transactions $t_{a,c}$ and transaction $t_{b,c}$ are assumed to
  fail.
\end{itemize}
Since the monitoring tree has height $2 < 2+1$, no decision needs to
be made about the first pending transaction, and the history remains
unchanged: $H_2 = H_1 = H_{0}$.

Finally, we apply the last transaction, $t_{c,\langle a, b \rangle}$,
to blockchain run $(H_2,\tau_2)$.
We show the resulting run in
Fig.~\ref{fig:traces}(d-\(\<step>((H_2,\tau_2),t_{c,\langle a, b
  \rangle})\)), and since it is the first time a monitoring window
ends, we also show the intermediate states in
Fig.~\ref{fig:traces}(d-\(\<extend>(\tau_2,t_{a,\langle a, b
  \rangle})\)) and (d-\(\<innerprune>(N^c_1)\)).
We execute transaction $t_{c,\langle a, b \rangle}$ in all four leaves
of $\tau_2$, resulting in one configuration in all cases, as no
contract monitors transaction $t_{c,\langle a, b \rangle}$.
There is only one configuration where contract \(c\) has enough tokens
to perform \(t_{c, \langle a , b \rangle}\), \(N^{\cc}_2\), and thus,
in its committing successor, $N^{\ccc}_3$, tokens are returned to $a$
and $b$.
At the other configurations, \(N^{\fc}_2,N^{\cf}_2\) and
\(N^{\ff}_2\), transaction \(t_{c, \langle a , b \rangle}\) fails, and
thus, all the remaining leaves have only a failing successor.
The intermediate monitoring tree resulting from attaching transaction
\(t_{c,\langle a,b \rangle}\) at \(\tau_{2}\), \(\tau'_2\), is shown
in
Fig.~\ref{fig:traces}(d-\(\<extend>(\tau_2,t_{c,\langle a, b
  \rangle})\)).

The monitoring window for transaction $t_{a,c}$ has ended and a
decision has to be made.
To that end, function \(\<step>\) invokes function \(\<prune>\) to
decide if the first pending transaction, \(t_{a,c}\), fails or
commits.
Function \(\<prune>\) first invokes function \(\<innerprune>\) to
remove all impossible nodes.
Function \(\<innerprune>\) is recursively applied to all subtrees.
Fig.~\ref{fig:traces}(d-\(\<innerprune>(N^c_1)\)) shows the result of
pruning the subtrees with root \(N^c_1\).
The future monitor for transaction at \(N^c_1\), \(t_{b,c}\), commits
if $b$ receives its token back before its monitoring window ends, which
happens in $N^{\ccc}_3$, the only future in $N^c_1$ committing
subtree.
Therefore, all nodes in the failing subtree of $N^c_1$ are impossible
and hence removed by function \(\<innerprune>\).
Similarly, the future monitor for transaction at \(N_0\), \(t_{a,c}\),
commits if $a$ receives its token back before its monitoring window
end, which happens in $N^{\ccc}_1$, the only possible future in $N_0$
committing subtree.
Hence, all nodes in the failing subtree of \(N_0\) are impossible and
removed by function \(\<innerprune>\).
Therefore, once the pruning is complete, node \(N_0\) has no failing
successor and its committing subtree is the pruned version of the
subtree rooted at \(N^c_1\).
Then, function \(\<prune>\) returns the pruned subtree at \(N^c_1\) as
the new monitoring tree, \(\tau_3\), committing transaction
\(t_{a,c}\).
Finally, function \(\<step>\) extends \(H_2\) by making transaction
\(t_{a,c}\) permanent to obtain the new history \(H_3\) (see
Fig.~\ref{fig:traces}(d-\(\<step>((H_2,\tau_2), t_{c,\langle a,b \rangle})\))).
Notice that if node \(N^{\cff}_3\) had been considered when deciding
to commit or fail transaction \(t_{a,c}\), its future monitor would
have incorrectly failed, as \(a\) token is not returned in the
configuration at node \(N^{\cff}_3\), but that node is impossible.

\newcommand{\trace}[2]{\includegraphics[width=#2\textwidth]{figures/trace-#1.pdf}}

\begin{figure}[t!]
  \footnotesize\centering
  \noindent
  \begin{tabular}{@{}c@{}c@{}c@{}}
    
      \trace{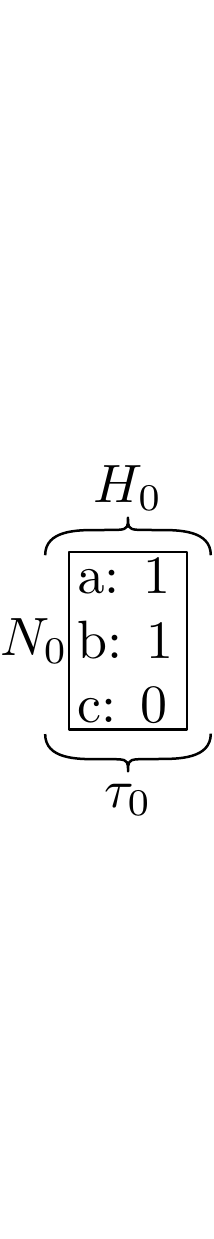}{0.07} \vline & \trace{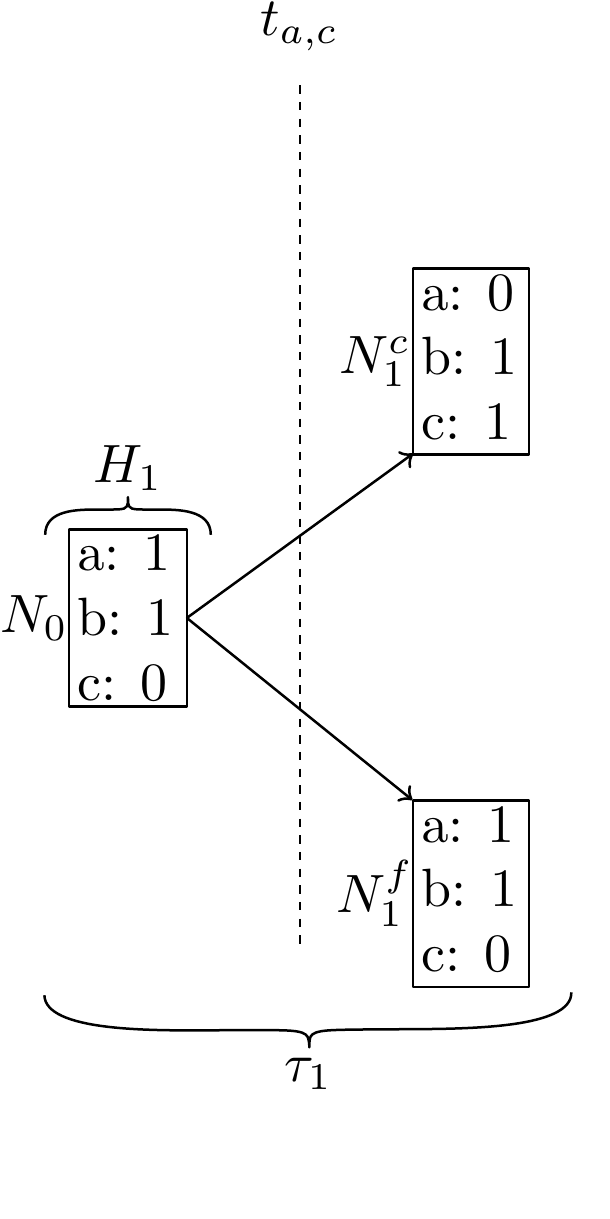}{0.2} & \vline \trace{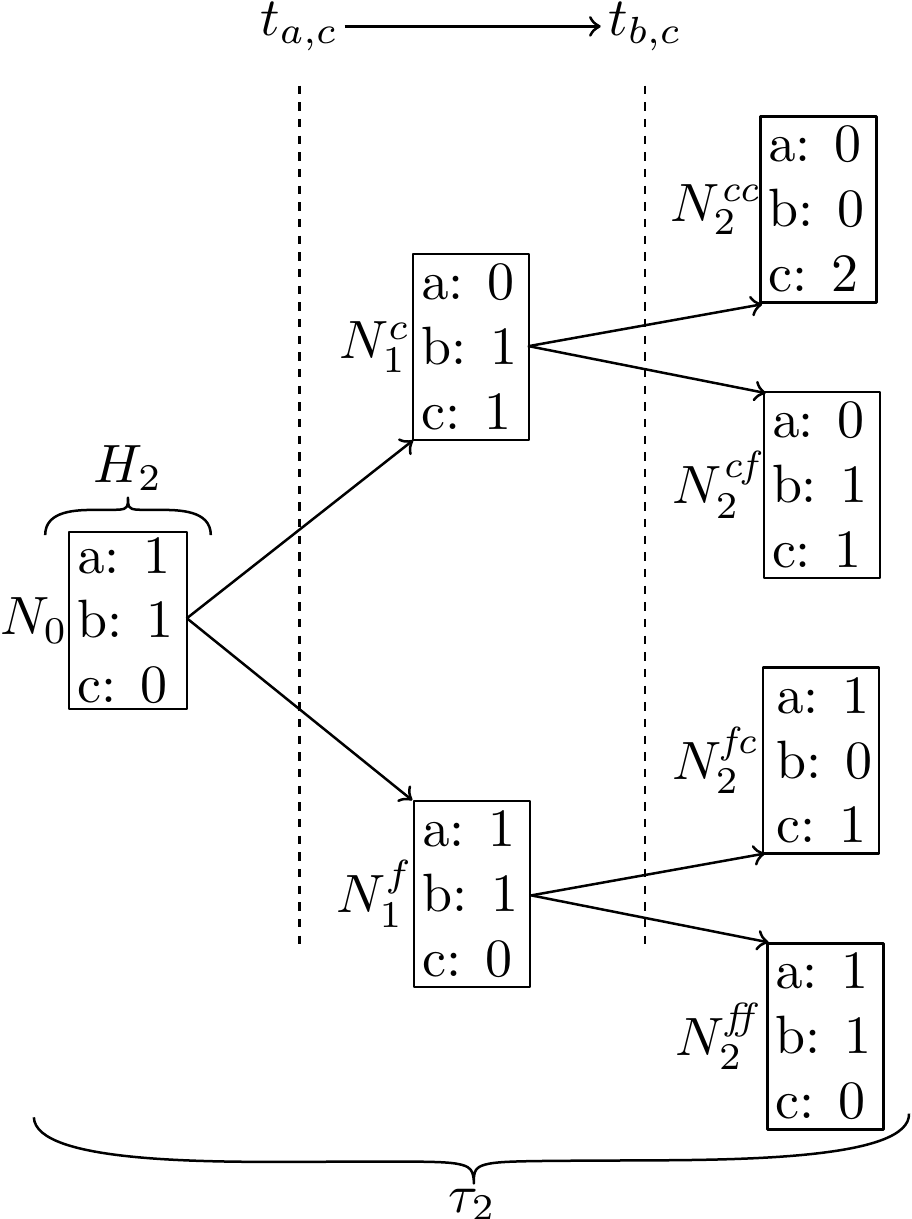}{0.3}\\
    (a)  & (b-\(\<step>((H_0,\tau_0),t_{a,c})\))  & (c-\(\<step>((H_1,\tau_1),t_{b,c})\)) \\
    
  \end{tabular}
  \begin{tabular}{@{}c@{}c@{}}
    \trace{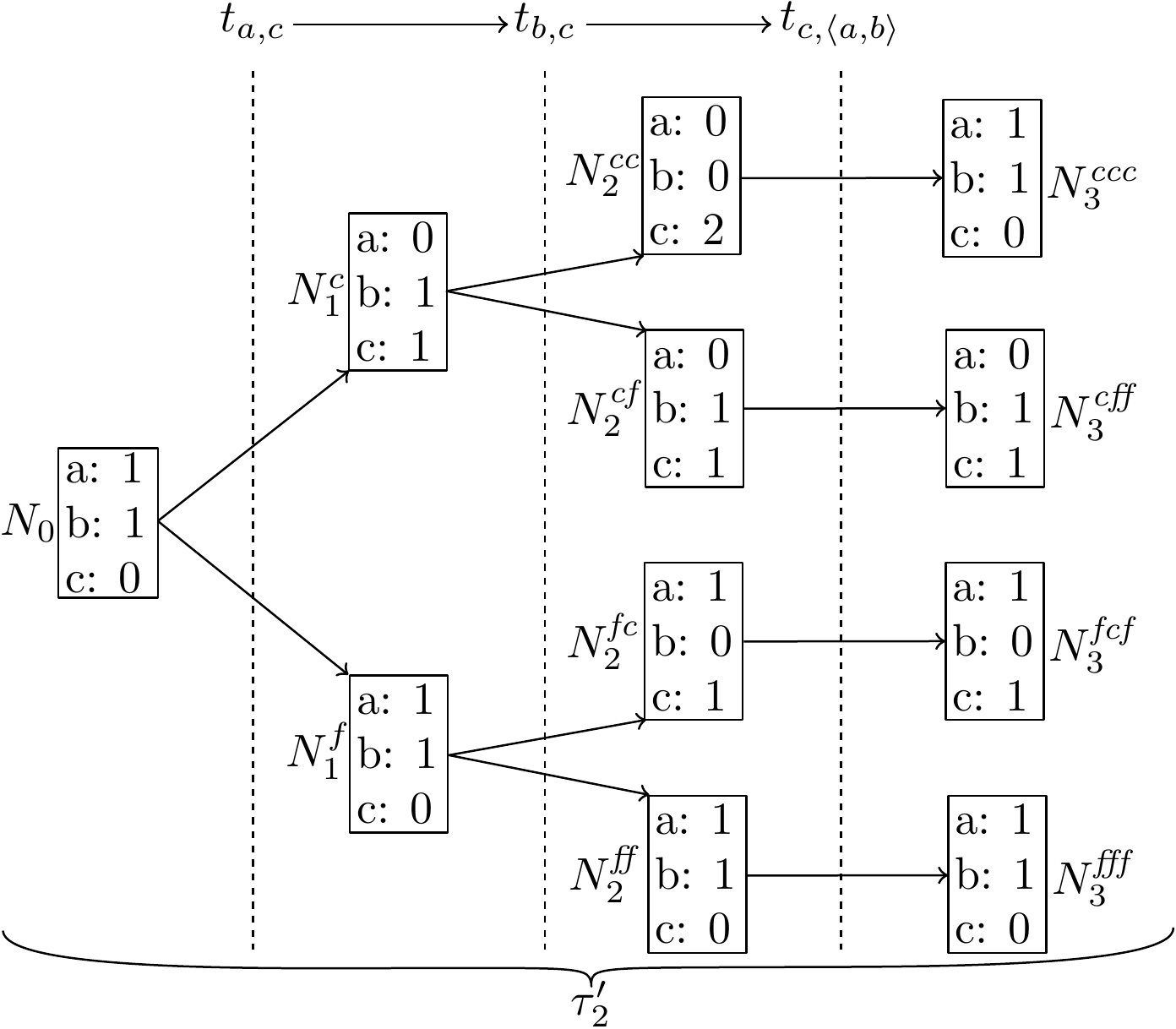}{0.4} & \vline \trace{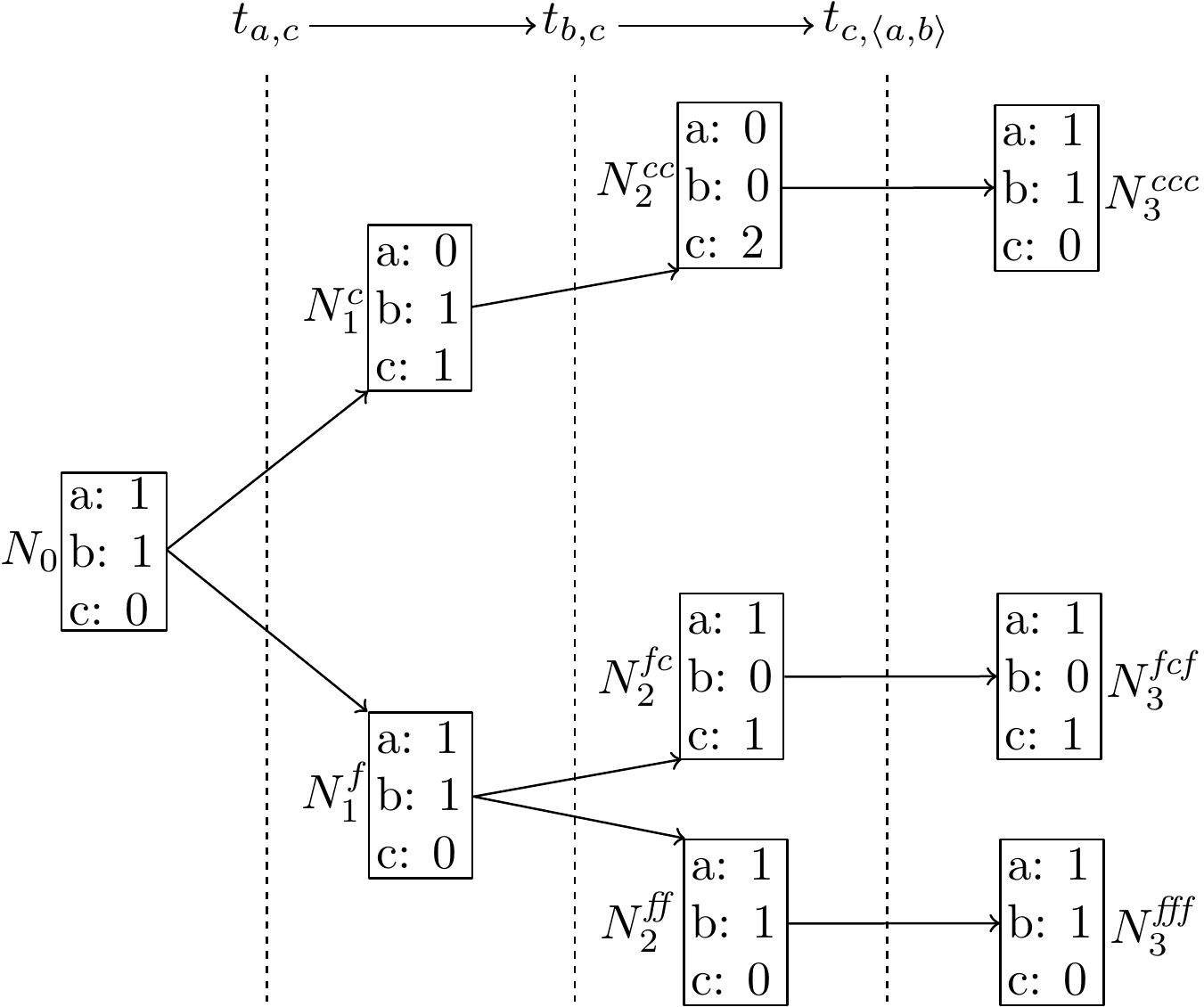}{0.42} \\
    (d-\(\<extend>(\tau_2,t_{c,\langle a,b\rangle})\)  & (d-\(\<innerprune>(N^c_1)\)) \\
  \end{tabular}
  \trace{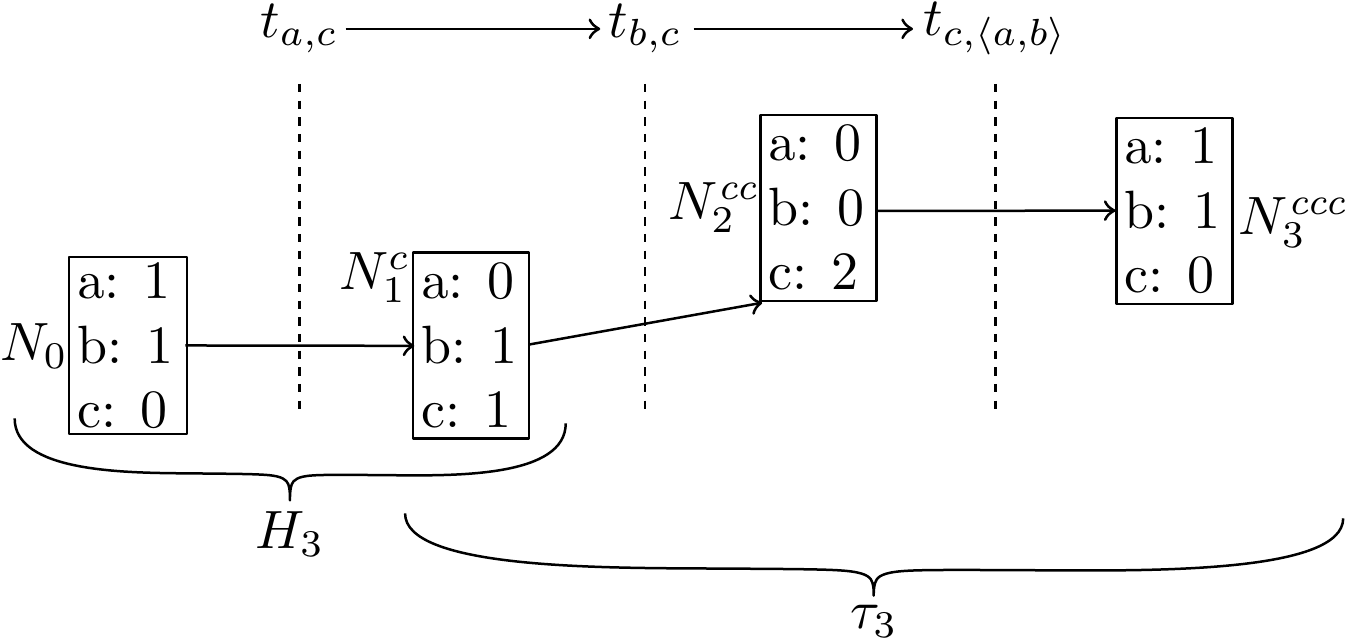}{0.4} \\
  (d-\(\<step>((H_2,\tau_2),t_{c,\langle a,b\rangle})\))
  \caption{Blockchain system evolution}
  \label{fig:traces} 
\end{figure}



\section{Properties Proofs}\label{sec:proofs}
In this appendix section, we prove the properties of the model
of computation presented in Section~\ref{sec:properties}.
We assume there is a fixed monitoring window \(k>0\).


The result of pruning a monitor tree is a monitor subtree of the original one
where we keep the root.
\begin{lemma}\label{lmm:prune}
  Let \(\tau\) and \(\nu\) be a monitoring tree such that \(\<innerprune>(\tau) = \nu\).
  Then, \(\nu\) is a subtree of \(\tau\) with the same root as \(\tau\) and all leaves
in \(\nu\) are also leaves in \(\tau\).
\end{lemma}

\begin{proof}
The proof is by structural induction on the tree structure of monitor tree
\(\tau\).
  
\textbf{Base Case:} If \(\tau\) is a leaf, then \(\nu = \tau\) and the
claim trivially holds.
  
\textbf{Inductive Step:}
We split the inductive step into two cases,
based on the number of successors of \(\tau\):
\begin{itemize}
  \item Monitor tree \(\tau\) has only one successor \(\tau'\):
    By inspecting the definition of \(\<innerprune>\), we can see that the
result of \(\<innerprune>(\tau)\) is a tree that has the same root as \(\tau\) and in
its only successor is the pruned version of \(\tau'\), \(\nu'\).
    By the inductive hypothesis, all leaves in \(\nu'\) are also
    leaves in \(\tau'\).
    Since all leaves in \(\tau'\) are also leaves in \(\tau\), it
    follows that all leaves in \(\nu\) are also leaves in \(\tau\).
    Finally, by inductive hypothesis, \(\nu'\) is a subtree of \(\tau'\)
    with the same root as \(\tau'\), making \(\nu\) a subtree of
    \(\tau\).
  \item If \(\tau\) has two successors, then there are three cases
    depending on if function \(\<innerprune>\) removes the failing
    subtree of \(\tau\), the committing subtree of \(\tau\) or neither.
    For all three cases, the proof is analogous to the previous
    case.
\end{itemize}
\end{proof}

Now, we prove each lemma stated in Section~\ref{sec:properties}.

\repeatatlemma{bounded-certainty}
\begin{proof}
  The proof is by induction on the number of steps taken by the step function,
\(l\):
  
  \textbf{Base Case:} If no transaction were added, \(l = 0\), then \(\tau\)
contains only one node, the initial configuration, and the lemma follows.
  
  \textbf{Inductive step:} Let \((H',\nu)\) be a blockchain run
  obtained by applying function \(\<step>\) \(l\) times, and \(t\) a
  transaction such that \(\<step>((H',\nu),t) = (H,\tau)\).
  Function \(\<step>\) starts by invoking function
  \(\<extend>(\nu,t)\), which returns a new monitoring tree, \(\eta\),
  obtained by adding one or two successors to each leaf in \(\nu\).
  Then, all leaves in \(\eta\) are added nodes.
  By inductive hypothesis, all leaves in \(\nu\) are in its last
  level.
  Then, all leaves in \(\eta\) are also its last level.
  
  Furthermore, by inductive hypothesis, the height of \(\nu\) is
  \(\min(k,l)\), hence the height of \(\eta\) is \(\min(k,l) + 1\).
  If \(l < k\) then \(\tau = \eta\) and the lemma follows. 
  Otherwise, the height of \(\eta\) is \(k+1\) and function
  \(\<step>\) invokes function \(\<prune>\), setting
  \(\tau = \<prune>(\eta)\).
  Function \(\<prune>\) returns a subtree rooted in one of the
  successors of \(\<innerprune>(\eta)\), thus the height of \(\tau\)
  is one less than the height of the pruned version of \(\eta\).
  Lemma~\ref{lmm:prune} directly implies that
  \(\<innerprune>(\eta)\) has the same height as \(\eta\).
  As consequence, the height of \(\tau\) is \((k+1)-1 = k\).
\end{proof}

\repeatatlemma{monitoring-tree}
\begin{proof}
Function \(\<step>((H,\tau),t)\) begins by invoking function
\(\<extend>(\tau,t)\) which returns a new monitoring tree, \(\nu\),
obtained by adding one or two successor to each leaf in \(\tau\).
Thus, the height of \(\nu\) is one more than the height of \(\tau\).
By lemma~\ref{l:bounded-certainty}, the height of \(\tau\) is \(k\),
so the height of \(\nu\) is \(k+1\).
Therefore, function \(\<step>\) extends history \(H\) by making
permanent the first transaction in \(\tau\) to obtain \(H'\), and
invokes function \(\<prune>\) to get the new monitoring tree, that is,
\(\tau' = \<prune>(\nu)\).
Function \(\<prune>(\nu)\) first invokes function
\(\<innerprune>(\nu)\), to obtain the pruned version of \(\nu\),
\(\eta\), and then functions \(\<prune>\) returns either the
committing or failing subtree of \(\eta\).
That is, \(\tau'\) is either the committing or failing subtree of the
pruned version of \(\nu\).
Since \(\eta\) is a subtree of \(\nu\) with the same root as \(\nu\)
(see Lemma~\ref{lmm:prune}), the root of \(\tau'\) is one of the
successors of the root of \(\nu\), which, by definition of \(\nu\) is
one of the successors of the root of \(\tau\).
Finally, from Lemma~\ref{lmm:prune}, all paths in \(\tau'\) are also
path in \(\nu\).
Specifically, all paths without leaves in \(\tau'\) are also path
without leaves in \(\nu\).
And, by definition of \(\nu\), all paths without leaves in \(\nu\) are
path in \(\tau\).
Consequently, all paths without leaves in \(\tau'\) are path in \(\tau\).
\end{proof}

\repeatatlemma{innerprune-subtree}
\begin{proof}
  The proof is by induction in \(\tau\).
  
  \textbf{Base Case:} If \(\tau\) is a leaf, then it is not
  impossible and \(\<innerprune>(\tau) = \tau\).

  \textbf{Inductive Step:} There are two cases, depending on the
  number of successors of \(\tau\):
  \begin{itemize}
  \item if \(\tau\) has only one successor, \(\eta\), then
    \(\<innerprune>(\tau)\) returns a tree, \(\tau'\) that has the same root as
    \(\tau\) and in its only successor is the pruned version of
    \(\eta\), \(\eta'\).
    By inductive hypothesis, \(\eta'\) does not have impossible nodes;
    thus, \(\tau'\) does not have impossible nodes either.
    Moreover, only the nodes removed when pruning \(\eta\) are
    removed from \(\tau\), and by inductive hypothesis, those nodes
    are impossible.
    Consequently, only impossible nodes are removed from \(\tau\).
  \item if \(\tau\) has two successors, \(\eta_c\) and \(\eta_f\),
    \(\<innerprune>\) computes the pruned version of each of them,
    \(\eta'_c\) and \(\eta'_f\), respectively.
    By inductive hypothesis, \(\eta'_c\) has no impossible nodes,
    thus, by inspecting the status of the future monitor of the
    transaction at the root of \(\tau\) in all leaves in \(\eta'_c\)
    one can know if there are impossible nodes in \(\tau\).
    If in all futures in \(\eta'_c\) the future monitor is known to
    commit, then all nodes in the failing subtree of \(\tau\) are
    impossible.
    In this case, function \(\<innerprune>\) returns a tree that has
    the same root as \(\tau\) and its only successor is
    \(\eta'_c\).
    Similarly, if in all futures in \(\eta'_c\) the future monitor is
    known to fail, then all nodes in the committing subtree of
    \(\tau\) are impossible, and function \(\<innerprune>\) returns a
    tree that has the same root as \(\tau\) and in its only successor
    is \(\eta'_f\).
    Otherwise, function \(\<innerprune>\) returns a tree that has the
    same root as \(\tau\) and in its successors are \(\eta'_c\) and
    \(\eta'_f\).
    In all cases, the lemma follows by inductive hypothesis.
\end{itemize}
\end{proof}

\repeatatlemma{prune}
\begin{proof}
It follows directly from Lemma~\ref{l:innerprune-subtree}, which
ensures that the pruned version of monitoring trees do not contain
impossible nodes, and the definition of function \(\<prune>\) (see
Fig.~\ref{fig:alg:funcs}).
\end{proof}

\repeatatlemma{size}
\begin{proof}
  By lemma~\ref{l:bounded-certainty} monitoring trees have height at
  most \(k\).
  For this proof we consider, w.l.o.g., monitoring trees with height
  exactly \(k\).
  
  By construction, nodes in monitoring trees have at most two
  successors.
  New successors are only added inside function \(\<extend>\).
  Specifically, \(\<extend>\) adds two successors to a node only if the
  corresponding transaction is being monitored.
  Thus, the \(m\) levels in the monitoring tree that correspond to
  monitored transactions can have at most double the number of nodes
  of the previous level.
  On the other hand, the \(k-m\) levels in the monitoring tree that
  correspond to unmonitored transactions have the same number
  of nodes as the previous level.
  Therefore, a monitoring tree has the maximal number nodes when
  monitored transactions generate branching in all nodes that are
  applied, and they precede all unmonitored transactions.
  In this case, the first \(m\) levels of the monitoring tree
  correspond to a complete binary tree, and the remaining \(k-m\)
  levels have the same number of nodes as the last level of that
  complete binary tree, which is \(2^{m}\).
  Consequently, the maximal size of monitoring trees with \(m\)
  monitored transactions is \(2^{m+1}-1 + 2^{m} \times (k-m)\) which
  is in \(\mathcal{O}(2^m \times k)\).
\end{proof}

\repeatatlemma{legacy-tx}
\begin{proof}
Let \(n\) be a node in
\(\tau\) that is not a leaf, such that its extended blockchain
configuration is \((\Sigma,\Delta,\mathcal{U})\) and
\(\<nextTx>(n) = t\).
Then, since \(\tau\) is a legacy monitoring tree, the function
\(\<applytx>\) can return two possible values:
\begin{itemize}
\item
  Transaction commits,
  \(\<applytx>(n,\<nextTx>(n)) = \commitType(\Sigma_c,\Delta_c,
  \mathcal{U}_c)\), producing a committing step
  \((\Sigma,\mathcal{U}) \leadsto_{t}(\Sigma_c, \mathcal{U}_c)\).
  As \(\<extend>\) is the only function that adds edges to monitoring
  trees, it follows that \(n\) has only one successor, whose extended
  blockchain configuration is \((\Sigma_c,\Delta_c, \mathcal{U}_c)\).
\item
  Transaction fails,
  \(\<applytx>(n,\<nextTx>(n)) = \failType(\Sigma,\Delta,
  \mathcal{U}_f)\), producing a failing step
  \((\Sigma,\mathcal{U}) \leadsto_{t}(\Sigma, \mathcal{U}_f)\).
  Similarly, in this case, node \(n\) also has only one successor, whose
  extended blockchain configuration is
  \((\Sigma,\Delta, \mathcal{U}_f)\).
\end{itemize}  

In consequence, monitoring tree \(\tau\) is a chain, and for each pair
of consecutive nodes \(n_i \xrightarrow{\<tx>} n_{i+1}\) with extended
blockchain configurations \((\Sigma_i,\Delta_i,\mathcal{U}_i)\) and
\((\Sigma_{i+1},\Delta_{i+1},\mathcal{U}_{i+1})\), respectively, the
relation
\((\Sigma_i,\mathcal{U}_i) \leadsto_{tx}
(\Sigma_{i+1},\mathcal{U}_{i+1}) \) holds.
This implies that the execution of transactions in \(\tau\) is
equivalent to executing them in the traditional model of computation.
\end{proof}

\repeatatlemma{legacy}
\begin{proof}
It follows from that every transaction in \(H\) coincides
with rule \(\leadsto\), the root configuration at \(\tau\) being
the last configuration in \(H\) and Lemma~\ref{l:legacy-tx}.
\end{proof}


\section{Atomic Loan - Example Explained}\label{app:Example}

In this section, we describe in detail the blockchain evolution of the example
presented in Section~\ref{sec:example}.
For simplicity, we neglect gas consumption.

\renewcommand{\L}{\texttt{L}\xspace}

Let \NC and \L be two contracts installed in a blockchain with
monitoring window of length \(2\), where \NC runs
\lstinline|NaiveClient| and \L runs \lstinline|Lender|.
Consider \((\Sigma_0,\Delta_0)\) to be a blockchain configuration
where \NC has 100 tokens and \L has 1000 tokens.
We describe the evolution of the system from blockchain run
\((H_0,\tau_0)\) with \( H_0 = \tau_0 = (\Sigma_0,\Delta_0)\), and
executing the following three transactions in order:
\begin{inparaenum}[(1)]
\item \(t_{req}\), \NC requests a loan for 100 to \L;
\item \(t_{inv}\), \NC invests if it has more than 200 tokens;
\item \(t_{ret}\), \NC returns the loan to \L.
\end{inparaenum}

Let \(N_0\) be the root of \(\tau_{0}\).
The execution of transaction \(t_{req}\) in blockchain run
\((H_0,\tau_0)\) starts by executing transaction $t_{req}$ in the only
leaf of \(\tau_0\), $N_0$.
Transaction \(t_{req}\) is originated by an external invocation to
function \lstinline|borrow|.

Within the execution of $t_{req}$, function \lstinline|loan| is
invoked.
Since lender \L has enough tokens, it will send the requested amount
to client \NC, and it will update its \lstinline|pending_returns|
with the lent amount, and also set the entry for transaction $t_{req}$
to undecided in its \lstinline|failmap| map.
As consequence, there is a branching in the monitoring tree.
In node \(N_0\)'s committing successor, \(N^c_1\), the loan is assumed to
take place and the balance of client \NC is increased by 100 tokens
while the balance of lender \L is decreased by 100 tokens.
On the other hand, in the failing successor, \(N^f_1\), transaction
\(t_{req}\) fails and \(N^f_1\) has the same configuration as \(N_0\)
There is no need to prune the new monitoring tree as its height is
$1 < k + 1$, and thus the history remains the same.

In the subsequent transaction, $t_{inv}$, client \NC try to invest its
tokens by invoking function \lstinline|invest|.
Transaction $t_{inv}$ executes in both leaves of the monitoring tree,
$N_1^c$ and $N_1^f$, and since it is not monitored, each of them has
one successor.
In configuration $N_1^c$, client \NC has enough tokens to invest,
resulting in a committing successor, $N_2^{\cc}$, where client NC
generated some profit.
In configuration $N_1^f$, client NC does not have enough tokens to
invest and therefore, it has only a failing successor, $N_2^{\ff}$,
whose blockchain configuration coincides with $N_1^f$.
Again, there is no need to prune the new monitoring tree as its height
is $2 < k + 1$.

Finally, the third transaction, $t_{ret}$, is originated by an invocation
to function \lstinline|payBack| in client \NC, to pay back the
100 tokens loan from transaction $t_{req}$.
Transaction $t_{req}$ is executed in configurations at $N_2^{\cc}$ and
$N_2^{\ff}$, the leaves of the new monitoring tree.
When executing transaction $t_{ret}$ in configuration $N^{\cc}$,
client \NC has enough tokens to send to lender \L, hence function
\lstinline|returnLoan| in lender \L is invoked with the identifier of
transaction \(t_{req}\) as parameter.
Lender \L receives the 100 tokens that it had lent in transaction
$t_{req}$, thus it sets the entry corresponding to $t_{ret}$ in its
\lstinline|failmap| map as \lstinline|COMMIT|.
Transaction \(t_{ret}\) is not monitored and thus \(N_2^{\cc}\) has
only one committing successor, \(N_3^{\ccc}\), where the balance of \NC
is decreased by 100 tokens and the one of \L is increased by 100.
Similarly, when executing transaction $t_{ret}$ in configuration
$N^{\ff}$, client \NC has enough tokens to send to lender \L and
function \lstinline|returnLoan| in lender L is also invoked with the
identifier of transaction \(t_{req}\) as parameter.
In this case, lender \L accepts the 100 tokens, but it does not update
its failing map because in this scenario transaction \(t_{req}\) have
failed, and thus, the value of pending return for transaction
\(t_{req}\) is set to \(-100 \neq 0\) during the execution of
\(t_{ret}\).
Transaction \(t_{ret}\) is not monitored and thus \(N_2^{\ff}\) has
only one committing successor, \(N_3^{\ffc}\).
As transaction \(t_{ret}\) is committed in this case, the balance of
\NC decreases by 100 tokens and the one of \L increase by 100, even
though \L has not lent any tokens to \NC in this scenario.

Let $\nu$ be the monitoring tree up until this point.
Fig.~\ref{fig:lender_client_alwayspays:simp} show the balance of smart
contract NC and L in the monitoring tree $\nu$.
Monitoring tree $\nu$ has height $3 = k + 1$, and thus the monitoring
window for the first transaction has ended, and a decision to commit
or fail it must be taken.
The first pending transaction in $\nu$ is $t_{req}$, where client \NC
requested the loan to lender \L.
The only leaf in $\nu$ that assumes that transaction $t_{ret}$ has
committed is $N_3^{\ccc}$, and in this leaf $t_{ret}$ is committed by the
monitor.
Consequently, $t_{req}$ is committed, and the subtree with root as the
committing successor of $N_0$, $N_1^c$, becomes the new monitoring
tree, discarding the failing subtree of \(N_0\), rooted at
\(N^f_1\), that assumes that transaction $t_{req}$ failed.

Even though in this run the configuration
$N_3^{\ffc}$ is removed, in that configuration, client \NC is paying
back a loan to lender \L that did not happen.


\end{document}